\documentclass[12pt,letterpaper]{article}
\usepackage[letterpaper]{geometry}
\usepackage[colorlinks=true, linkcolor=black, 
  citecolor=blue, urlcolor=blue]{hyperref}
\usepackage[round,authoryear]{natbib}
\usepackage{amssymb,amsmath,mathtools}
\usepackage{amsthm}
\usepackage{multirow}
\usepackage{float}
\usepackage{mathrsfs}
\usepackage{multicol}
\usepackage{anyfontsize} 
\usepackage{silence} 
\usepackage{graphicx}
\usepackage{color}

\graphicspath{{Figures/}}

\newtheorem{theorem}{Theorem}
\newtheorem{Lemma}{Lemma}
\newtheorem{corollary}{Corollary}
\theoremstyle{definition}
\newtheorem{definition}{Definition}

\theoremstyle{remark}

\DeclareMathOperator{\arctan2}{arctan2}
\DeclarePairedDelimiter{\ceil}{\lceil}{\rceil}

\newcommand{\Median}{\textnormal{Median}}

\newcommand{\CV}{\textnormal{CV}}
\newcommand{\CSD}{\textnormal{CSD}}
\newcommand{\MLE}{\textnormal{MLE}}
\newcommand{\CMAD}{\textnormal{CMAD}}
\newcommand{\CLMS}{\textnormal{CLMS}}
\newcommand{\CLTS}{\textnormal{CLTS}}

\newcommand{\uMLE}{\scalebox{0.65}{\textnormal{MLE}}}
\newcommand{\uCMAD}{\scalebox{0.65}{\textnormal{CMAD}}}
\newcommand{\uCLMS}{\scalebox{0.65}{\textnormal{CLMS}}}
\newcommand{\uCLTS}{\scalebox{0.65}{\textnormal{CLTS}}}
\newcommand{\uS}{\scalebox{0.65}{\textnormal{S}}}
\newcommand{\uimp}{\scalebox{0.65}{\textnormal{imp}}}
\newcommand{\uexp}{\scalebox{0.65}{\textnormal{exp}}}
\newcommand{\AVar}{\operatorname{AVar}}

\newcommand{\IF}{\textnormal{IF}}
\newcommand{\eps}{\varepsilon}
\newcommand{\Fe}{F_{\varepsilon}}
\newcommand{\q}{q} 
\newcommand{\y}{\theta} 

\newcommand{\hkappa}{\widehat{\kappa}}
\newcommand{\hmu}{\widehat{\mu}}
\newcommand{\hpsi}{\widehat{\psi}}
\newcommand{\hsigma}{\widehat{\sigma}}

\definecolor{blue}{RGB}{0,0,255}
\definecolor{red}{RGB}{255,0,0}

\begin{document}

\def\spacingset#1{\renewcommand{\baselinestretch}%
{#1}\small\normalsize} \spacingset{1}

\title{Robust measures of dispersion for circular data \\ 
with an anomaly detection rule}

\author{
  Houyem Demni\thanks{Department of Economics and Law, 
  University of Cassino and Southern Lazio, Italy,
  houyem.demni@unicas.it} 
  \and Mia Hubert\thanks{Section of Statistics and 
  Data Science, KU Leuven, Leuven, Belgium,
  mia.hubert@kuleuven.be}
  \and Giovanni C. Porzio\thanks{Department of Economics
  and Law, University of Cassino and Southern Lazio,
  Italy, porzio@unicas.it}
  \and Peter J. Rousseeuw\thanks{Section of Statistics 
  and Data Science, KU Leuven, Leuven, Belgium,
  peter@rousseeuw.net}
}
\date{February 28, 2026}
\maketitle

\begin{abstract}
Circular variables that represent directions or periodic observations arise in many fields, such as biology and environmental sciences. An important issue when dealing with circular data is how to estimate their dispersion robustly, avoiding undue effects of anomalies. This work extends three robust dispersion measures from the line to the circle. Their robustness is studied via their influence functions and relative bias curves. From these dispersion measures, robust estimators of parameters of circular distributions can be derived. This yields robust estimators for the concentration parameter of the von Mises distribution and the dispersion parameter of the wrapped normal distribution. Their breakdown values and statistical efficiencies are obtained, and they are compared in a simulation study. Building on the best performing estimator, a robust circular anomaly detection procedure is developed, and employed to visualize outliers through a circular violin plot. Three real datasets are analyzed.
\end{abstract}


\spacingset{1.2}

\section{Introduction}
\label{sec:intro}

In several application areas one may encounter 
circular data, consisting of angles, directions, 
and periodic observations that are represented 
as points on a circle. Circular data are found 
in biosciences with a broad range of 
interesting applications in several areas such 
as distributional models \citep{Mardia2007}, 
parametric and nonparametric estimation 
\citep{Tishkovskaya2021}, regression 
\citep{Hassanzadeh2021} and clustering 
\citep{Ranalli2020}. See also 
\cite{Ley2017,Ley2019} for a variety of modern 
methods and data examples. Because circular 
data lie on a bounded nonlinear manifold, 
specific statistical methods have been 
developed for their analysis, see for 
instance \cite{Buttarazzi2018}, 
\cite{Fernandez2007}, and \cite{DiMarzio2022}. 

Standard measures of location and dispersion 
include the circular mean and the circular 
standard deviation $\CSD$. Their precise 
definitions are given in 
Section~\ref{sect:back}. These measures are 
however very sensitive to outlying values. 
To illustrate this, we consider the directions 
in which 22 Sardinian sea stars traveled after 
displacement from their natural habitat 
\citep{Fisher1995}. The data are shown in 
Figure~\ref{fig:SeaStars} and contain two 
unusual observations that we colored red. 
The CSD of the full data set is 0.58, but when 
the two special observations are removed it 
drops substantially, to 0.36. 

\begin{figure}[ht]
\centering
\includegraphics[width=0.40\textwidth, trim=0cm 2cm 0cm 2cm,clip]{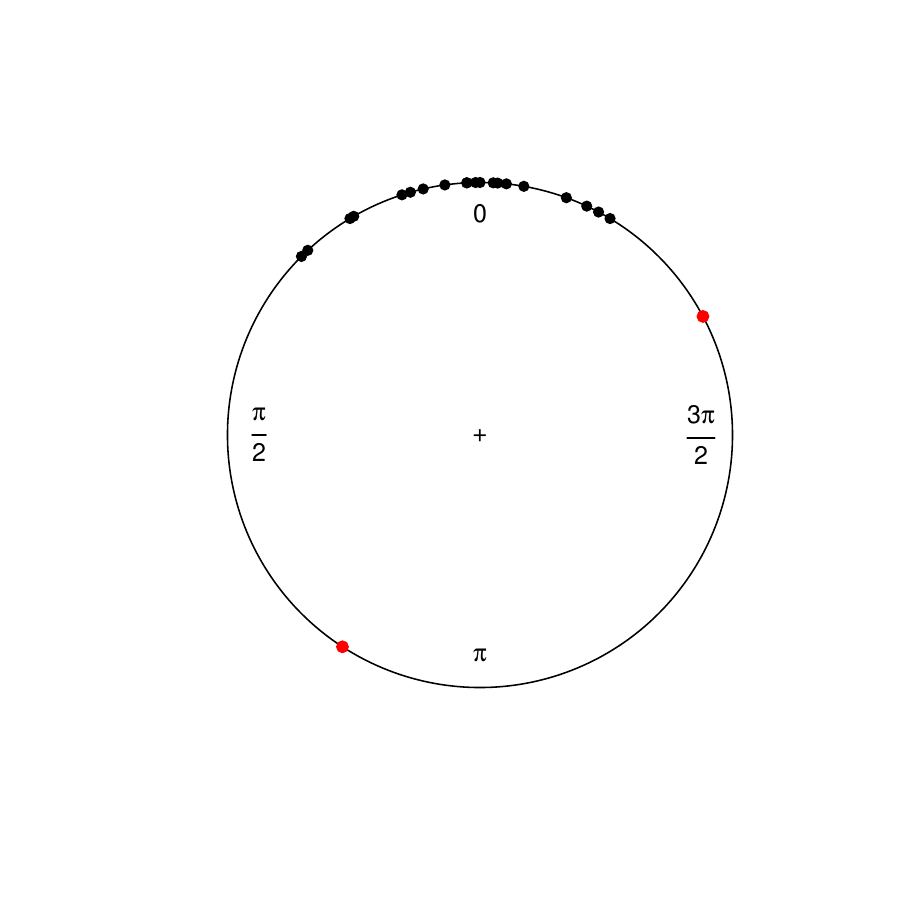}
\vspace{-3mm}
\caption{Circular plot of the directions traveled 
by Sardinian sea stars. Two unusual observations 
are colored red.}
\label{fig:SeaStars}
\end{figure}

To deal with this sensitivity to anomalous 
values, we can employ robust measures instead. 
A robust measure of location for circular data
is the \textit{circular median}. It has been 
shown to possess well-established robustness 
properties in the family of unimodal symmetric 
distributions \citep{Ko1988}. The circular 
median is the same for the full sea stars 
dataset and without the two suspicious 
observations.

However, robust measures of circular dispersion 
have received less attention so far, unlike in
the linear setting. A well-known robust 
dispersion estimator on the line is the median
absolute deviation (MAD), described by
\cite{Hampel1986} and \cite{Leys2013}. A circular
version called CMAD has been used to compare the 
performance of some circular location estimators
\citep{Otieno2006}, to lower a false alarm rate 
in the detection of transiting planets 
\citep{Seader2013}, and to obtain a smoothing 
parameter when nonparametrically estimating the 
off-pulse interval of a pulsar light curve 
\citep{Schutte2016}. 

In spite of this, the CMAD has not yet been
studied by the tools of robust statistics
such as the influence function and the
breakdown value. We will do so in this paper.
We will also study the circular counterparts 
of two other robust measures of scale for 
linear data, the least median squares (LMS)
and least trimmed squares (LTS) methods, that
were first proposed by \cite{LMS} in a 
regression setting. The resulting CLMS and
CLTS methods will be defined and compared
with CMAD.

We will also use these dispersion measures
to derive robust estimators for the dispersion 
parameters of standard circular distributions. 
We will provide a detailed treatment of the 
estimation of the concentration parameter 
$\kappa$ of the von Mises distribution, and of
the dispersion parameter $\sigma$ of the 
wrapped normal distribution.

Finally, we construct an outlier detection
rule on the circle based on our results, and
introduce a companion graphical tool that is
a variation on the well-known violin plot. 

The paper is organized as follows. 
Section~\ref{sect:back} reviews a few 
background concepts of circular data.
The $\CMAD$, $\CLMS$ and $\CLTS$ measures
are defined in Section~\ref{sect:disp}.
Their influence functions are derived in 
Section~\ref{sect:IF}, and their relative bias 
under contamination is computed in 
Section~\ref{sect:bias}. 
Section~\ref{sect:conc} constructs 
corresponding estimators for dispersion 
parameters of circular distributions, 
provides their breakdown values, and studies
their efficiency. Section~\ref{sect:violin} 
introduces a robust anomaly detection rule 
and the circular violin plot. 
Section~\ref{sect:data} analyzes three real 
data examples, and Section~\ref{sect:concl} 
concludes.

\section{Background material} \label{sect:back}

Circular data consists of points on the unit 
circle. The same point can be either defined 
in polar coordinates by an angle 
$\theta \in [-\pi, \pi)$ measured in radians, 
once an origin and an orientation have been 
chosen, or in Cartesian coordinates as a vector 
$x = (x_1,x_2)^T = (\cos \theta, \sin \theta)^T$
so $x \in {\mathcal S}:= \{ x: ||x||= 1 \}$.
Back-transformation is obtained through 
$\theta = arctan2(x_2,x_1)$.

Note that the angles $\theta$ and $\theta+2\pi$ 
yield the same point on the circle. The arithmetic 
for circular data is therefore modulo $2\pi$, and 
for the sake of simplicity we will write 
$\theta_1 + \theta_2$ instead of 
$\theta_1 + \theta_2 \mod 2\pi$ where $\theta_1$ 
and $\theta_2$ are two points on the circle.

The notation based on angles is more economical 
than the notation with Cartesian coordinates, so 
we will mostly use angles. A circular random 
variable can be denoted as $\Theta$, and its 
distribution function as 
$F(\theta)=P(\Theta \leqslant \theta)$ for 
$\theta \in [-\pi, \pi)$. 

Throughout this work we will repeatedly use the 
length of the shortest arc joining two points on 
the unit circle, that we denote as the shortest 
arc distance $d(\theta_1, \theta_2) := 
\pi - |\pi - |\theta_1 - \theta_2| |$.

\subsection{Location and dispersion of circular 
   random variables} 

The location of a circular random variable is 
often measured by its circular mean. Let $\Theta$ 
be a circular random variable with distribution 
$F$. Its circular mean is defined as the angle
$\mu(F)=\arctan2(\beta_1,\alpha_1)$ 
with \[
\alpha_1 = \mathbb{E}(\cos \Theta) = 
\int_{-\pi}^{\pi} \cos\theta \, dF(\theta) 
\qquad \mbox{ and } \qquad
\beta_1 = \mathbb{E}(\sin \Theta) = 
\int_{-\pi}^{\pi} \sin\theta \, dF(\theta).
\]

The \textit{circular median set} of $\Theta$ is 
instead defined as the set of points that minimize 
the mean angular distance. 

\begin{definition}[Circular Median Set]
The circular median set $M_{\Theta}$ of a random 
variable $\Theta$ is defined as:
\begin{equation}\label{eq:medianSetDefinition}
  M_{\Theta} := \{\theta: 
  \underset{\theta}{{\arg\min} } \hspace{2pt}
  \mathbb{E}(d(\Theta,\theta)) \}, \hspace{0.5cm}
  {\theta \in [-\pi,\pi)}
\end{equation}
where $ \mathbb{E}(T)$ is the expected value of 
the (linear) random variable $T$. 
\end{definition}

Any element in $M_{\Theta}$ is called a 
\textit{circular median} or Fr\'echet median. 
The circular median set can be a singleton, the 
whole circle (for instance, at the uniform 
distribution) or even a disconnected set. 
Circular random variables with unique median do 
exist. For instance, if a distribution admits a 
continuous density and it is unimodal, then its 
circular median is unique \citep{Mardia1972}.
If $\Theta$ has a unique circular median, we 
denote it as $\tilde{\mu}$.

\begin{definition}[Circular Variance]
The circular variance of a random variable 
$\Theta$ is given by
\begin{equation} \label{eq:circvar}
  \CV(F) = 2 (1 - \rho)= \int_{-\pi}^{\pi} 
  2(1-\cos(\theta -\mu)) dF(\theta) 
\end{equation}
where $\mu$ is the circular mean of 
$\Theta$ and
\begin{equation}
\rho = \left| \int_{-\pi}^{\pi} e^{i\theta} 
  \, dF(\theta) \right|
= \int_{-\pi}^{\pi} \cos(\theta - \mu)\,
  dF(\theta)
\end{equation}
is the mean resultant length. 
\end{definition}
Both the circular variance and the mean 
resultant length are bounded, with 
$0 \leqslant \CV(F) \leqslant 2$ and $0 
\leqslant \rho \leqslant 1$. Some authors 
drop the factor 2 in~\eqref{eq:circvar}, but 
we prefer this version because for increasing
concentration, when the distribution is 
concentrated on a small near-linear part of 
the circle,~\eqref{eq:circvar} corresponds 
to the linear variance. This is because the 
leading term of the Taylor expansion of 
$2(1-\cos(\theta-\mu))$ for $\theta$ around 
$\mu$ is $(\theta-\mu)^2$, see also page 21
of \cite{Jammalamadaka2001}.
Here we will mainly work with the square root 
of~\eqref{eq:circvar}, which is the 
\textit{circular standard deviation} denoted 
by $$\CSD(F) = \sqrt{\CV(F)}\,.$$

A measure on the circle is said to be rotationally 
invariant if it is unchanged under all rotations. 
A measure is equivariant under rotations if 
rotating the circle rotates the measure in the
same way. A circular distribution $F$ is said to 
be reflectionally symmetric about a direction 
$\mu$ if its probability function is such that 
$P(\Theta \in (\mu - \theta, \mu)) =  
P(\Theta \in (\mu, \mu + \theta))$ for all 
$\theta \in [-\pi, \pi)$ . 

If a circular distribution $F$ admits a density 
$f$, it is said to be rotationally symmetric of 
order $m$, with $m \in \mathbb{N}$, if the density 
repeats every $2\pi/m$ radians, that is, 
$f(\theta) = f(\theta + 2\pi/m)$ for all $\theta$. 
It is said to be antipodally symmetric if 
$f(\theta) = f(\theta + \pi)$ for all 
$\theta \in [-\pi, \pi)$, that is, if the density 
is invariant under a rotation of $\pi$ radians. 
Moreover, the distribution is monotone 
reflectionally symmetric about $\mu$ if it is
reflectionally symmetric and $f(\mu+\theta)$
is non-increasing for $\theta \in [0,\pi]$, so 
that the density attains its maximum at $\mu$ 
and decreases symmetrically in both directions.

\subsection{The von Mises and wrapped normal
            distributions}

The von Mises distribution $vM(\mu,\kappa)$ is 
the most often used distribution to model 
circular data, and sometimes called the circular 
normal distribution. Its density is given by
\begin{equation*}
   f_{vM}(\theta; \mu, \kappa) :=
   \frac{1}{2 \pi I_0(\kappa)} 
   \exp (\kappa \cos (\theta-\mu)),
\end{equation*}
with $I_0$ the modified Bessel function of the 
first kind and order $0$. Here $\mu \in [-\pi,\pi)$ 
is the location parameter and $\kappa \geqslant 0$ 
is the concentration parameter. The location 
parameter $\mu$ is simultaneously the circular mean, 
the mode and the circular median. The density has 
a shape resembling the normal distribution, and it 
approaches the normal distribution as its 
concentration $\kappa$ gets higher.
For $\kappa=0$, the distribution reduces to the 
continuous uniform distribution on the circle. 
For $\kappa > 0$, the distribution is unimodal 
and reflectionally symmetric around $\mu$. 
The concentration parameter $\kappa$ and the 
circular variance are related through the mean 
resultant length 
$\rho = I_1(\kappa)/I_0(\kappa)$ with $I_1$ 
the modified Bessel function of the first kind of 
order $1$. Consequently, $\kappa$ is implicitly 
defined as the solution of 
$I_1(\kappa)/I_0(\kappa) = 1 - \CSD^2/2$, which 
admits no analytic expression and must be 
evaluated numerically.

Another widely used reflectionally symmetric 
distribution for unimodal circular data is the 
wrapped normal distribution $WN(\mu,\sigma^2)$ 
with density
\begin{equation*}
  f_{WN}(\theta; \mu, \sigma^2) := 
  \frac{1}{\sigma \sqrt{2 \pi }}  
  \overset{\infty}{\underset{m=-\infty}{\sum}} 
  \exp \Big(\frac{-(\theta-\mu+2 \pi {m})^2}
  {2 \sigma^2} \Big) \, .
\end{equation*}
Its location parameter is $\mu \in [-\pi,\pi)$, 
and $0<\sigma^2<+\infty$ is the variance of the 
linear normal distribution that gets wrapped 
around the circle.
The mean resultant length of the wrapped normal 
distribution is $\rho = \exp(-{\sigma^2}/{2})$, 
which decreases as $\sigma^2$ increases. The CSD 
is directly related to $\sigma^2$ by 
$\CSD^2 = 2(1 - \exp(-\sigma^2/2))$.

\section{Three measures of circular dispersion} 
\label{sect:disp}
In this section we define the Circular Median 
Absolute Deviation (\CMAD), the Circular Least 
Median Spread (\CLMS), and the Circular Least 
Trimmed Standard Deviation (\CLTS).

\subsection{The Circular Median Absolute Deviation
(CMAD)}

The CMAD is defined for distributions $F$ with a
unique circular median and equals the (linear) 
median of the distribution of the shortest arc 
distances from this (circular) median.

\begin{definition}[CMAD]
Let ${\mathcal M}_1$ be the set of circular random 
variables whose unique circular median is 
$\tilde{\mu}$. The CMAD of 
$\Theta \in {\mathcal M}_1$ is given by:
\begin{equation}\label{eq:CMAD}
  \CMAD(F) := \Median(d (\Theta, \tilde{\mu}))
\end{equation}
where $\Median(T)$ refers to the median of the 
(linear) random variable $T$. That is, $\Median(T)$ 
is the set of values $x$ such that 
$P(T \leqslant x) \geqslant 0.5$ and 
$P(T \geqslant x) \geqslant 0.5$. If such a set is 
an interval, its midpoint is taken as the median 
value.
\end{definition}

For circular random variables admitting a density, 
$\CMAD(F)$ is half the length of the shortest arc 
that is symmetric around $\tilde{\mu}$ and 
encompasses probability 0.5:
\begin{align}
\label{eq:CMAD_density}
  \CMAD(F)& = \frac{1}{2} \hspace{3pt} 
   d(\theta_1, \theta_2), \\
  \{ \theta_1, \theta_2 \}  & =  
  \underset{\theta_1,\theta_2 \in [-\pi,\pi)}
  {\arg \min}\hspace{2pt}
  \{d(\theta_1, \theta_2) \hspace{2pt} | 
  \hspace{2pt} P(\theta_1 \leqslant \Theta 
  \leqslant \theta_2) = 0.5, \hspace{2pt} 
  d(\theta_1, \tilde{\mu})=d(\theta_2, 
  \tilde{\mu})\} \nonumber
\end{align}

The Circular Median Absolute Deviation ($\CMAD$) 
is nonnegative. It is rotational invariant because
the circular median is rotational equivariant and 
the shortest arc distances are rotational invariant. 
It is also reflectionally invariant because the 
circular median is equivariant under reflection, 
while the shortest arc distances are invariant.

CMAD is a bounded measure of circular dispersion 
that satisfies $0 \leqslant \CMAD(F) < \pi$.
We now investigate the range that CMAD can take,
depending on the distribution. 
The proof of the next lemma, as well as the proofs 
of all other theoretical results, is deferred 
to the Supplementary Material.

\begin{Lemma}[Values of $\CMAD$]\label{Lemma:CMAD=0}
Let $\Theta \in {\mathcal M}_1$ be a circular 
random variable, and $\theta \in [-\pi, \pi)$. 
\begin{align*}
 & \textit{If} \hspace{6pt} \exists \hspace{2pt}
   \theta : P(\Theta = \theta) > 0.5,
   \hspace{12pt} then \hspace{6pt} \CMAD(F)=0\,;\\
 & \textit{If} \hspace{6pt} \exists \hspace{2pt} 
   \theta : P(\Theta = \theta) = 0.5, 
   \hspace{12pt} then \hspace{6pt} 0 \leqslant 
   \CMAD(F) < \pi/2 \,;\\
 & \textit{If} \hspace{6pt} \forall \hspace{2pt} 
   \theta : P(\Theta = \theta) < 0.5, \hspace{12pt}
   then \hspace{6pt} 0 < \CMAD(F) < \pi\,.
\end{align*}
\end{Lemma}

By Lemma~\ref{Lemma:CMAD=0}, $\pi$ is an upper bound
for $\CMAD$. The following result shows that this
bound is sharp.

\begin{theorem}\label{Lemma:CMAD_Ubound}
For $F$ ranging over the set of cdf's of all
circular random variables, $\sup_F \CMAD(F) = \pi$\,.
\end{theorem}

For reflectionally symmetric distributions, the value 
of $\CMAD$ is a quantile of the distribution itself. 
The following result follows directly from 
Equation~\eqref{eq:CMAD_density}. 

\begin{theorem}[Value of $\CMAD$ under 
reflectional symmetry]\label{Th:CMAD_value_s}
Let $\Theta \in {\mathcal M}_1$ be 
reflectionally symmetric about some $\mu^*$.
Then the circular median of $\Theta$ is either 
$\mu^*$ or $\mu^* + \pi$. If $\mu^*$, $\CMAD(F)$ 
equals the length of the shortest arc joining 
$\mu^*$ with the point $\xi$ cumulating 25\% of 
probability from $\mu^*$ itself, either 
clockwise or counterclockwise. That is, with
\begin{equation}\label{eq:xi}
\xi: \int_{\mu^*}^{\xi}  dF(\theta) = 0.25,
\end{equation}
we have:
\begin{equation}\label{eq:CMADsym}
    \CMAD(F) = d (\xi, \mu^*).
\end{equation}
If the circular median of $\Theta$ is 
$\mu^*+\pi$, Equations~\eqref{eq:xi} 
and~\eqref{eq:CMADsym} apply to $\mu^*+\pi$.
\end{theorem}

The CMAD is not defined when the circular median is
not unique. For this reason we define two other
dispersion measures that do not depend on the
circular median.

\subsection{The Circular Least Median Spread (CLMS)}

The Circular Least Median Spread ($\CLMS$)
extends the Least Median Squares dispersion
measure \citep{LMS,Rousseeuw1987} from the 
line to the circle. 
The CLMS is the smallest median distance of 
the random variable to a point on the circle.

\begin{definition}[CLMS]
Let $\Theta$ be a circular random variable with
distribution $F$, and $\theta$ ranges over all 
points on the circle. Then the $\CLMS$ is 
given by:
\begin{equation*}
  \CLMS(F):= \hspace{2pt} 
  \underset{{\theta} }{\min} \hspace{2pt}
  {\Median}  (d (\Theta,\theta)).
\end{equation*}
\end{definition}

Equivalently, CLMS is half the length of the 
shortest arc that encompasses at least half of 
the probability distribution:

\begin{equation}\label{eq:CLMSDefinition2}
  \CLMS(F)=\frac{1}{2} \min\{d(\theta_1, \theta_2)
  | P(\theta_1 \leqslant \Theta \leqslant \theta_2)
  \geqslant 0.5\},
\end{equation}
where $\theta_1$ and $\theta_2$ are two points on the
circle. The $\CLMS$ is nonnegative, and invariant 
under rotation and reflection. We now study the 
conditions under which the $\CLMS$ reaches its 
minimum and maximum attainable value. 

\begin{Lemma}[Conditions for $\CLMS=0$]
\label{Lemma:CLMS=0}
Let $\Theta$ be a circular random variable with 
distribution $F$, and $\Delta_{\theta}$ denote
the point mass circular distribution in $\theta$. 
We have:
$$\CLMS(F) = 0 \hspace{12pt} \textit{iff} 
 \hspace{12pt} F = (1-\varepsilon) 
 \Delta_{\theta}+ \varepsilon G, \hspace{12pt}
 0 \leqslant \varepsilon \leqslant 0.5,$$
for any value $\theta$ and any distribution G.
\end{Lemma}

\begin{Lemma}[Maximum of $\CLMS$]
\label{Lemma:CLMS_max}
Let $\Theta \in [-\pi, \pi)$ be a circular random 
variable with distribution $F$. Then,
$$\max_{F} \; \CLMS(F) = \pi/2.$$
\end{Lemma}

The result holds because $\CLMS(F)$ cannot be 
larger than $\pi/2$, since $\pi$ is the maximal 
length of the shortest path connecting two points
on a circle. On the other hand, $\CLMS=\pi/2$
characterizes a class of distributions, as seen 
in the following theorem.

\begin{theorem}
\label{Th:CLMS_Characterization}
Let $\Theta$ be a circular random variable with 
distribution $F$. Then
$$\CLMS(F) = \pi/2 \hspace{12pt} 
  \Leftrightarrow
  \hspace{12pt} F \in H_{0.5}$$
where $H_{0.5}$ is the class of distributions that
satisfy $\int_{HC} dH_{0.5}(\theta) = 0.5$ for any 
semicircle $HC$.
\end{theorem}

The class $H_{0.5}$ is a subset of the class of 
the centrally (or antipodal) symmetric 
distributions $H$, because they satisfy 
$\int_{HC} dH(\theta) \geqslant 0.5$ for any 
semicircle $HC$, see e.g.\ the discussion in 
\cite{Nagy2024}.

Note that the CSD has a similar implication
to the $\Leftarrow$ in
Theorem~\ref{Th:CLMS_Characterization}.
Because distributions in $H_{0.5}$ have antipodal 
symmetry their resultant length is 0, so
their CSD takes the maximal value $\sqrt{2}$.
But the $\Rightarrow$ does not hold. For 
instance, a rotationally symmetric distribution 
$F$ of order $m=3$ has $CSD(F) = \sqrt{2}$, but 
it does not belong to $H_{0.5}$.

A second remark contrasts $\CLMS$ with $\CMAD$.
By definition, the latter is not defined under 
$H_{0.5}$ since those distributions do not have 
a unique circular median.

On the other hand, $\CLMS$ takes the same value 
as $\CMAD$ under monotone rotational symmetry, 
as stated below.

\begin{theorem}\label{Th:value_CLMS_rs}
Let $\Theta$ be a continuous circular random 
variable whose distribution $F$ admits a density 
$f$ that is reflectively symmetric about some 
$\mu^*$ and strictly unimodal with mode $\mu^*$. 
Let  $\xi$ be the point cumulating 25\% of
probability from $\mu$ itself, as in 
\eqref{eq:xi}. Then
\begin{equation*}\label{eq:CLMSValue}
  \CLMS(F) = d(\xi,\mu^*). 
\end{equation*}
\end{theorem}

The theorem holds because the shortest arc is 
symmetric about the modal direction $\mu^*$. 
Moreover, by the continuity and the probability 
constraint in Equation~\eqref{eq:CLMSDefinition2}, 
half of its length has a probability mass of 0.25. 
Note that monotone reflectional symmetry is a 
special case of reflectional symmetry, under which
condition the same value was obtained for $\CMAD$ 
in Theorem~\ref{Th:CMAD_value_s}.

\subsection{The Circular Least Trimmed Standard 
deviation (CLTS)}

The third robust measure of dispersion we discuss 
is the Circular Least Trimmed Standard deviation
$\CLTS$, which extends the least trimmed squares
(LTS) dispersion \citep{LMS} from the line to the 
circle. The $\CLTS$ is given by the smallest 
$\CSD$ over arcs that encompass at least 50\% of 
probability mass.

For any $\theta_1$ and $\theta_2$, and any circular 
random variable $\Theta$ with distribution $F$, 
denote by $F_{[\theta_1,\theta_2]}$ the conditional 
distribution of $\Theta$ given 
$\theta_1 \leqslant \Theta \leqslant \theta_2$ and 
$\mu_{[\theta_1,\theta_2]}$ its circular mean.

\begin{definition}[CLTS]
Let $\Theta$ be a circular random variable with
distribution $F$. Then the Circular Least Trimmed 
Standard deviation (CLTS) is defined as
\begin{equation}\label{eq:CLTSDefinition}
  \CLTS(F):= \underset{{\theta_1, \theta_2} }
  {\min}\, \{ \CSD(F_{[\theta_1,\theta_2]})
  \;|\;  P(\theta_1 \leqslant \Theta \leqslant
  \theta_2) \geqslant 0.5 \}.
\end{equation}
\end{definition}

The $\CLTS$ is nonnegative, rotational and 
reflective invariant because the CSD is.
From~\eqref{eq:circvar} it follows that 
\begin{equation} \label{eq:CLTSformula}
\CLTS(F) := \underset{{\theta_1, \theta_2} }
  {\min}\,
\left\{\sqrt{\frac{\int_{\theta_1}^{\theta_2}
2(1-\cos(\theta-\mu_{[\theta_1,\theta_2]}))
dF(\theta)} {P(\theta_1 \leqslant \Theta 
 \leqslant \theta_2)}}\;|\; 
P(\theta_1 \leqslant \Theta \leqslant \theta_2)
\geqslant 0.5 \right\}\,.
\end{equation}

\begin{Lemma}[Conditions for $\CLTS=0$]
\label{Lemma:CLTS=0}
Let $\Theta$ be a circular random variable 
with distribution $F$, and $\Delta_{\theta}$ be 
the point mass in some $\theta$. Then
$$\CLTS(F) = 0 \hspace{12pt} \Leftrightarrow 
\hspace{12pt} F = (1-\varepsilon) 
\Delta_{\theta}+ \varepsilon G, \hspace{12pt}
0 \leqslant \varepsilon \leqslant 0.5,$$
for any value $\theta$ and any distribution G.
\end{Lemma}

For the maximal value of CLTS, we can rely on 
the following result.

\begin{Lemma}\label{Lemma:CLTS_Unif}
Let F be the continuous circular uniform distribution. 
That is, F has density $1/(2\pi)$. Then
$$\CLTS(F) = \sqrt{2(1 - 2 /\pi)}.$$
\end{Lemma}
The result can be derived from~\eqref{eq:CLTSformula}, 
because each semicircle contains exactly 0.5 of the 
probability mass. We conjecture that there is no other 
circular distribution (either discrete or continuous) 
whose $\CLTS$ attains this value.

Finally, it is possible to derive a general equation 
for the $\CLTS$ in the case of monotone reflective 
symmetry.

\begin{theorem}\label{Th:value_CLTS_rs}
Let $\Theta$ be a continuous circular random variable
whose distribution $F$ admits a density $f$ that is 
reflectively symmetric about some $\mu^*$ and 
strictly unimodal with mode $\mu^*$. Let $\xi$ be the 
point that cumulates 25\% of probability from $\mu^*$
as in~\eqref{eq:xi}. Then
\begin{equation}\label{eq:CLTSValue}
\CLTS(F)=\sqrt{4\int_{\mu^*}^{\mu^*+\xi}
 2(1-\cos(\theta - \mu^*))dF(\theta)
 } \,.
\end{equation}
\end{theorem}
This holds by~\eqref{eq:CLTSformula} and 
because the conditional distribution with minimal
CSD is reflectively symmetric about $\mu^*$.

\subsection{Values, infimum, and supremum at von 
  Mises and wrapped normal}

\begin{corollary}\label{corollary:CMAD_value_vM}
Let $F =vM(\mu, \kappa >0)$ or 
$F = WN(\mu,\sigma^2)$ and denote
$q_1 =  F^{-1}(1/4)$ and $q_3=F^{-1}(3/4)$. Then
\begin{eqnarray}
\CMAD(F) =\CLMS(F) = q_3 - \mu 
  \label{eq:CMAD_vonMises} \\
\CLTS(F) = \CSD(F_{[q_1,q_3]}) 
  \label{eq:CLTS_vonMises}
\end{eqnarray}
\end{corollary}
This follows from Theorems~\ref{Th:CMAD_value_s}, 
\ref{Th:value_CLMS_rs} and~\ref{Th:value_CLTS_rs}.
The infimum and the supremum of $\CMAD$, $\CLMS$, 
$\CLTS$ at a von Mises distribution with 
$\kappa>0$ and at the wrapped normal are 
given below.

\begin{corollary}\label{Lemma:inf_sup_vM}
Let $\Theta \sim vM(\mu, \kappa >0)$ or
$\Theta \sim WN(\mu,\sigma^2)$. Then
\begin{eqnarray*}
  \inf \; \CMAD(F) = \inf \; \CLMS(F) = 
  \inf \; \CLTS(F) = 0 \\
  \sup \; \CMAD(F) = \sup \; \CLMS(F) = \pi/2 \\
  \sup \; \CLTS(F) = \sqrt{(2(1-2/\pi)}\,.
\end{eqnarray*}
\end{corollary}

At the von Mises distribution, the infimum holds 
because of Lemma~\ref{Lemma:CMAD=0}, 
Lemma~\ref{Lemma:CLMS=0}, and
Lemma~\ref{Lemma:CLTS=0}, and because 
$F \stackrel{d}\rightarrow \Delta_\mu$ for 
$\kappa \rightarrow + \infty$.
The supremum follows from~\eqref{eq:CMAD_vonMises} 
and Lemma~\ref{Lemma:CLTS_Unif}. It holds that
$\lim_{\kappa \rightarrow 0} \CMAD(F)= 
\lim_{\kappa \rightarrow 0} \CLMS(F)=
\frac{\pi}{2}$ and $\lim_{\kappa \rightarrow 0}
\CLTS(F) = \sqrt{(2(1-2/\pi)}$.

At the wrapped normal distribution we have 
$\lim_{\sigma \rightarrow 0} \CMAD(F)=
\lim_{\sigma \rightarrow 0} \CLMS(F)=
\lim_{\sigma \rightarrow 0} \CLTS(F)=0$
because for $\sigma \rightarrow 0$ the 
distribution becomes increasingly concentrated
around $\mu^*$. We also have
$\lim_{\sigma \rightarrow + \infty} \CMAD(F)=
\lim_{\sigma \rightarrow +\infty} \CLMS(F)=
\frac{\pi}{2}$, and 
$\lim_{\sigma \rightarrow +\infty}\CLTS(F)=
\sqrt{(2(1-2/\pi)}$, because for 
$\sigma \rightarrow +\infty$ the distribution 
approaches the uniform distribution. 

\section{Influence functions} \label{sect:IF}

The robustness of a functional can be measured 
in different ways. First, we consider its 
influence function which evaluates the effect 
of a small fraction of contamination. Formally, 
the influence function \citep{Hampel1986} 
of a functional $S$ in a point $\y$ at a model 
distribution $F$ is defined as 
\begin{equation}\label{eq:defIF}
 \IF(\y;S,F) := \frac{\partial}{\partial \eps}
 S\big((1-\eps)F + \eps\Delta_\y\big)
 \Big|_{\eps=0+} =
 \lim_{\eps \downarrow 0} \frac
 {S\big((1-\eps)F + \eps\Delta_\y\big) -
 S(F)}{\eps}\;.
\end{equation}

The influence function of CSD has been provided 
in \cite{Ko1988} and is given by
\begin{equation*}
\IF(\theta;\CSD,F) = \; \frac{1}{\CSD(F)}\Big(
  1 - \frac{\CSD^2(F)}{2} - \cos(\theta)\Big) 
\end{equation*}
for $-\pi \leqslant \theta \leqslant \pi$, and 
$F$ a reflectively symmetric circular 
distribution with mode 0. 
We will now compute the influence functions of 
the $\CMAD$, $\CLMS$, and $\CLTS$ functionals.
In order to deal with contaminated distributions
like $(1-\eps)F + \eps\Delta_\y$ we will
consider distributions of the type
\begin{equation}\label{eq:Ptype}
 P = a_0P_0 + \sum_{j=1}^k a_j\Delta_{\y_j}
\end{equation}
where $a_0 > 0$, $P_0$ is a distribution with 
continuous density $f$ on the circle,
$k < \infty$, all $a_j > 0$, all $\y_j$ in the
interval $]-\pi,\pi]$, and $\Delta_{\y}$ the
probability distribution that puts all its mass
in the point $\y$. Since $P$ is a probability
it must hold that 
$a_0 + a_1 + \ldots + a_k = 1$. 
We then define the
cumulative distribution function $G$ by 
$G(\y) = \int_{-\pi}^\y g(u)du$ for 
$-\pi \leqslant \y \leqslant \pi$\,. We thus 
have $G(-\pi) = 0$ and $G(\pi) = 1$. 
We then extend $G$ to the interval 
$[-2\pi,2\pi]$ by $G(\y) = G(\y + 2\pi) -1$ 
for $-2\pi \leqslant \y < -\pi$
and $G(\y) = G(\y-2\pi)+1$ for 
$\pi < \y \leqslant 2\pi$. $G$ is right 
continuous so we can define $G^{-1}$ in the
usual way.

For any distribution function of the 
type~\eqref{eq:Ptype} we can now write the 
functional $\CLMS(G)$ as
\begin{equation}\label{eq:functCLMS}
 \CLMS(G) = \frac{1}{2}\inf \Big\{ G^{-1}\Big(
 t+\frac{1}{2}\Big) - G^{-1}\Big(t\Big)\;;\;\;
  G(0)-1 < t \leqslant G(0)+\frac{1}{2} \Big\}\,.
\end{equation}

\vspace{2mm}
\begin{theorem}[IF of $\CMAD$ 
and $\CLMS$]\label{Th.IFofCMADandCLMS}
Let $F$ be a circular distribution with continuous
density $f(\y) > 0$ that is reflectively symmetric
about 0 and strictly unimodal with mode 0 (i.e., 
$f$ is strictly decreasing on $[0,\pi]$). Then the 
influence functions of both $\CMAD$ and $\CLMS$ 
are given by
\begin{equation}
  \IF(\y; S, F) = \frac{\mbox{sign}(|\y|
  - F^{-1}(3/4))}{4 f(F^{-1}(3/4))}
\end{equation}
for $-\pi \leqslant \y \leqslant \pi$\,, where 
$\CMAD(F) = \CLMS(F) = F^{-1}(3/4)$ (by 
Theorems~\ref{Th:CMAD_value_s} 
and~\ref{Th:value_CLMS_rs}).
\end{theorem}

See Appendix B for the proof. Note that the 
conditions on $F$ are satisfied by most of 
the classic models of circular statistics, 
including any von Mises distribution (with 
$\mu = 0$ and $\kappa > 0$) and any wrapped 
normal distribution (with $\mu=0$). It is 
interesting that $\CMAD$ and $\CLMS$ have the
same influence function at such distributions
$F$, even though they are quite different 
estimators. 

Figure~\ref{fig:IFofCLMS} shows 
the IF of $\CMAD$ and $\CLMS$ for a von Mises 
distribution with $\mu = 0$ and $\kappa =1/2$. 
Outside the interval $[-\pi,\pi]$ it continues 
with period $2\pi$.

\begin{figure}[!ht]
\centering
\includegraphics[width=.70\textwidth]
{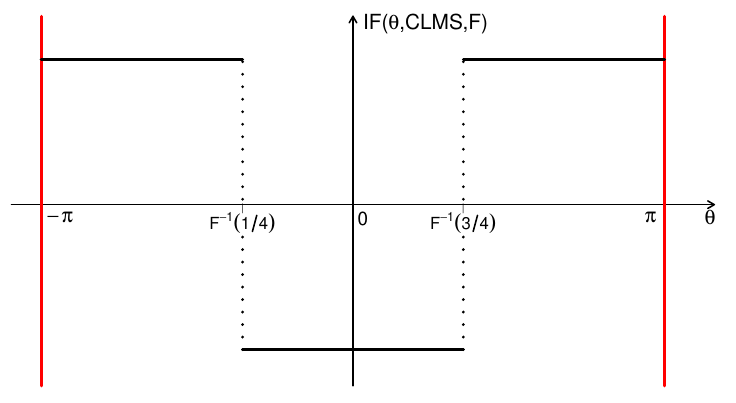}
\caption{Influence function of $\CMAD$ and
$\CLMS$ at the von Mises distribution with 
$\mu=0$ and $\kappa = 1/2$.}
\label{fig:IFofCLMS}
\end{figure}

We see that the IF is a piecewise constant
function which takes on the value 
$1/(4 f(F^{-1}(3/4)))$ for points outside
the interval that determines the scale for 
the uncontaminated distribution $F$, and 
the opposite value when it lies inside that
interval. This means that the standardized 
effect of a single point is quite small. 
When we add a point outside the interval
$\CMAD$ and $\CLMS$ only go up a little bit, 
and when we put it inside they only go down 
a little bit. The size of the effect is 
measured by the so-called 
\textit{gross-error-sensitivity} given by
$$\gamma^*(S,F) = \sup_\y |\IF(\y;S,F)|$$
so $\gamma^*(\CMAD,F) = 
1/(4f(F^{-1}(3/4)))=\gamma^*(\CLMS,F)$.
The influence function of $\CLTS$ looks 
different.

\vspace{2mm}
\begin{theorem}[IF of $\CLTS$]\label{Th.IFofCLTS}
Let $F$ be a circular distribution with continuous
density $f(\y) > 0$ that is reflectively symmetric 
about 0 and strictly unimodal with 
mode 0, that is, $f$ is strictly decreasing on 
$[0,\pi]$. Then the influence function of the 
$\CLTS$ estimator $S$ is given on $[-\pi,\pi]$ by
\begin{align}
  \IF(\y; S, F) =&\; \frac{1}{S(F)}\Big[
  1 - \frac{S^2(F)}{2} - \cos(F^{-1}(3/4)) 
  \nonumber \\ 
  &\qquad \;\;\;+ 2I\big(|\y| < F^{-1}(3/4)\big)
  \big(\cos(F^{-1}(3/4)) - \cos(\y)\big)
  \Big]\,.
\end{align}
\end{theorem}

Figure~\ref{fig:IFofCLTS} shows the IF of 
$\CLTS$ at the von Mises distribution with 
$\mu = 0$ and $\kappa = 1/2$. 
Note that the central part of the IF has the
shape of $-\cos(\y)$ whereas the IF of the 
LTS scale estimator in the linear setting
is quadratic in the center \citep{Croux1992}. 
However, when the 
von Mises distribution gets more concentrated 
($\kappa \rightarrow \infty$) the IF of $\CLTS$
becomes indistinguishable from that of the
linear LTS at a Gaussian distribution, since
the Taylor expansion of $-\cos(\y)$
around zero is $-1 + \y^2/2 + o(\y^3)$.
Also the IF's of $\CMAD$ and $\CLMS$ become 
indistinguishable from those of the scale
estimators MAD and LMS in the linear 
setting.

\begin{figure}[!ht]
\centering
\includegraphics[width=0.70\textwidth]
{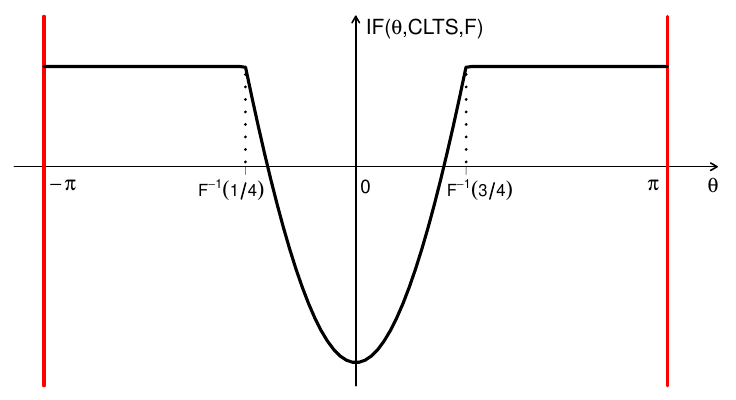}
\caption{Influence function of $\CLTS$ at the 
von Mises distribution with $\mu=0$ and
$\kappa = 1/2$.}
\label{fig:IFofCLTS}
\end{figure}

\vspace{2mm}
Note that Theorems~\ref{Th.IFofCMADandCLMS} 
and~\ref{Th.IFofCLTS} give the influence 
functions of the `raw' versions of $\CMAD$, 
$\CLMS$, and \CLTS. In the sequel, we will consider 
transformations $\zeta$ to make them consistent
estimators of a parameter in a parametric model, 
so the actual estimators will be of the form 
$\zeta(S)$ where $S$ is the raw estimator. In that 
situation, we can compute the IF of the final 
estimator by the chain rule:
\begin{equation}\label{eq:IFtransf}
 \IF(\y;\zeta(S),F) =
 \frac{\partial}{\partial \eps}
 \zeta\Big(S((1-\eps)F + \eps\Delta_\y)\Big)
 \Big|_{\eps=0+} 
 = \zeta'\big(S(F)\big) \IF(\y;S,F)\,,
\end{equation}
so it is proportional to the IF of the raw
estimator.

\section{Bias curves} \label{sect:bias}
Whereas the influence function measures the effect 
of a very small proportion of point contamination 
on a functional, we can also study the impact of 
other types and sizes of contamination.  

Here we consider von Mises distributions $F$  and 
contaminating distributions $G_{\theta,\eps} = 
(1-\varepsilon) F + \varepsilon \Delta_\theta$
for $\theta \in [-\pi,\pi)$ and 
$0 \leqslant \eps < 0.5$. We use point contamination 
because it often has the biggest effect on scale 
functionals \citep{Martin1989, Croux2001}.
We define the \textit{relative bias} in $\theta$ as
\begin{equation} \label{eq:relbias}
  rb(\theta; S, F, \eps):=
  \frac{S(G_{\theta,\eps})- S(F)}{S(F)}
\end{equation}
where $S$ is the CSD, CMAD, CLMS or CLTS functional. 
The denominator in~\eqref{eq:relbias} makes the 
values for the different functionals more comparable. 
Plotting the relative bias over the whole range of 
$\theta$ values yields the 
\textit{relative bias curve}.

To approximate these bias curves, we compute them on 
samples of size 5000. The sample versions of all the 
dispersion measures are computed using functions 
implemented in \textsf{R}. The sample $\CMAD$ is 
obtained by considering the empirical counterpart of 
Equation~\eqref{eq:CMAD} and is computed with the
\textsf{R} package \texttt{circular} 
\citep{Lund2013}. According 
to~\eqref{eq:CLMSDefinition2}, the sample CLMS 
is given by half the length of the shortest arc 
encompassing $h=\ceil{n/2}+1$ sample points, where 
$\ceil{\cdot}$ is the ceiling function. To obtain 
the sample $\CLTS$ we compute the CSD of all 
contiguous subsamples of size  $h=\ceil{n/2}+1$, 
and keep the lowest CSD.

Figure~\ref{fig:rbias_measures} shows the relative 
bias curves at von Mises distributions $F$ with 
location $\mu=0$ and concentration 
$\kappa \in \{2,5\}$. In each scenario, the 
contamination level was set to 
$\eps \in \{5\%,10\%,20\%\}$.

\begin{figure}[!ht]
\centering
\includegraphics[width=0.99\textwidth]
  {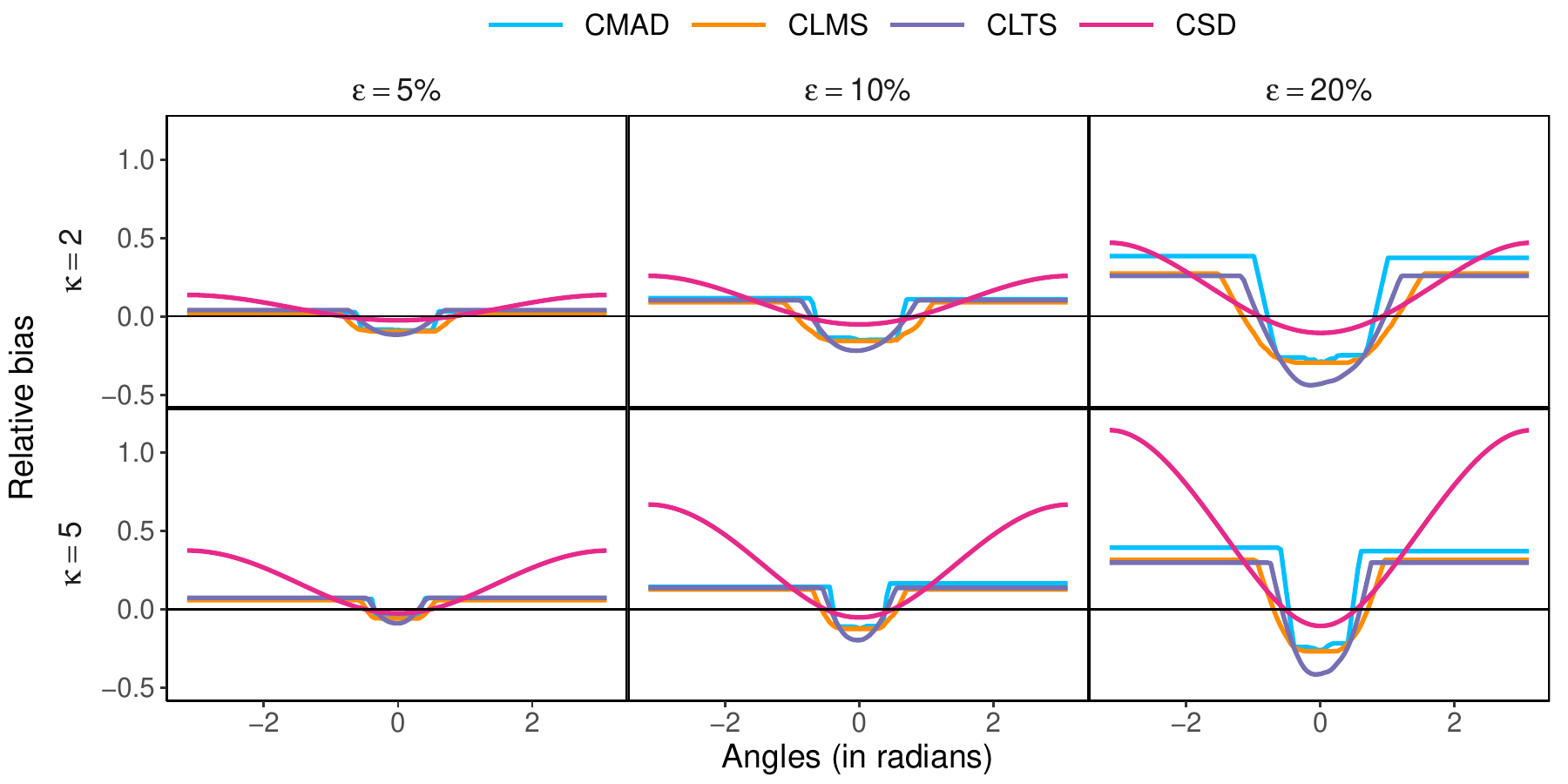}
\caption{Relative bias curves of $\CMAD$, $\CLMS$, 
$\CLTS$ and $\CSD$ at von Mises distributions with 
$\mu=0$. The rows are for concentration 
$\kappa = 2,5$. The columns have 
contamination rate $\eps= 5\%, 10\%, 20\%$.}
\label{fig:rbias_measures}
\end{figure}

We see that the bias curves of the dispersion measures 
have a shape that resembles their influence functions. 
When the contamination is located in the central part 
of the distribution, the bias is negative. 
Outside that region the bias is positive. 
The effect becomes larger with increasing contamination. 
The impact of the contamination on the circular 
standard deviation $\CSD$ is much larger than the 
effect on $\CMAD$, $\CLMS$ and $\CLTS$. Among the three 
robust estimators, CMAD shows the largest positive bias 
in the outer part of the distribution whereas CLTS has 
the largest negative bias in the central part, but the 
differences are slight.  
 
The value of $\kappa$ does not have a large impact on 
the size of the relative bias of the three robust 
measures. On the other hand, the CSD exhibits a much 
larger relative bias at higher $\kappa$ values. This 
can be understood as follows. When $\kappa$ is high, 
the uncontaminated data are tightly clustered around 
$\mu$, and the dispersion measure $\CSD$ is very 
precise with low variability. Then even a small 
amount of contamination can cause a noticeable 
relative change. In other words, when $\kappa$ is
high the CSD behaves like the classical standard 
deviation on the straight line.

These results are in line with those of \cite{Ko1988}.
They consider a standardized version of the influence
function, which is the influence function divided by
$S(F)$. In~\eqref{eq:relbias} we standardize the bias
in the same way. In their Theorem 4 they show that 
the standardized influence function of the classical
dispersion estimator becomes unbounded as $\kappa$
increases, and they conclude that this estimator is
not robust.

Our results indicate that, from a relative bias point 
of view, the three proposed robust measures perform 
much better than CSD. Overall, CLMS has the smallest 
bias.

We performed a similar analysis under the mean shift 
outlier model, where the uncontaminated and 
contaminating distributions have equal values of
$\kappa$ and differ only in location. The results
(not shown) are qualitatively similar to those in
Figure~\ref{fig:rbias_measures}.

\section{Robust estimation of parameters of circular 
  distributions}\label{sect:conc}

The proposed dispersion measures can be used to 
obtain robust estimators of parameters of 
circular distributions. We will illustrate this for 
von Mises distributions with concentration 
parameter $\kappa$, and for wrapped normal 
distributions with dispersion parameter $\sigma$.

We first look at the values of $\CMAD$, $\CLMS$, and
$\CLTS$ as a function of the population $\CSD$ in 
Figure~\ref{fig:relationship_CSD}, that were obtained 
numerically. The dashed horizontal lines indicate 
their infimum and their supremum. 

\begin{figure}[!ht]
\centering
\includegraphics[width=\textwidth,
   trim=0cm 1cm 0cm 1cm, clip]
   {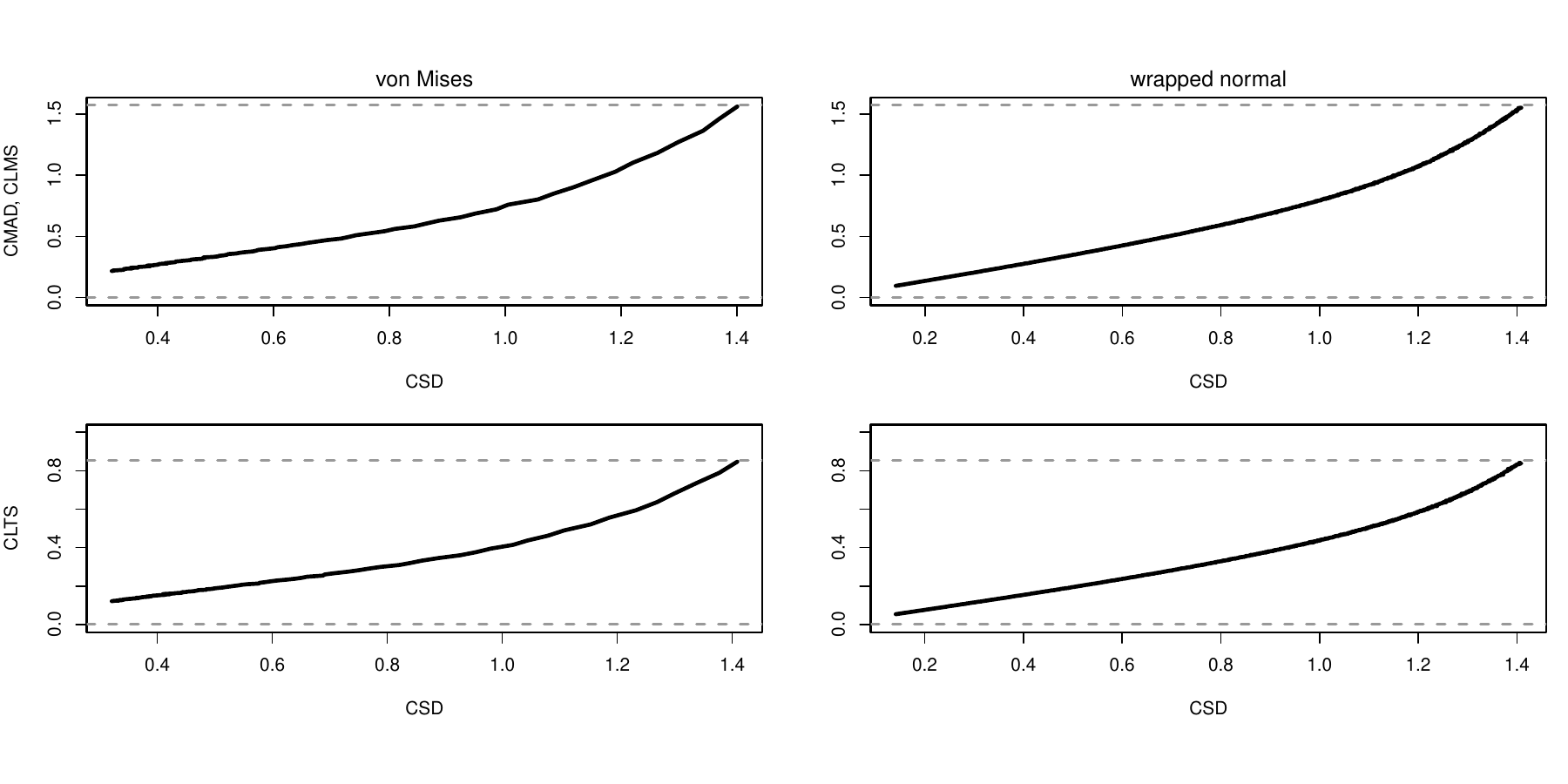}  
\caption{Values of the $\CMAD$, $\CLMS$, $\CLTS$ 
as a function of the population $\CSD$ at the 
von Mises distribution (left) with 
$0 < \kappa \leqslant 10$, and at the wrapped
normal distribution with 
$0 < \sigma \leqslant 3$. Dashed lines correspond 
to the infimum and to the supremum of these 
circular measures.}
\label{fig:relationship_CSD}
\end{figure}

Moreover, the population $\CSD$ is itself a
nonlinear function of the concentration 
parameter $\kappa$ of the von Mises distribution, 
and of the $\sigma$ of the wrapped normal 
distribution, as seen in
Figure~\ref{fig:CSD_kappa_Sigma} where the values 
of $\CSD$ are given as functions of $\kappa$
for $0 < \kappa \leqslant 10$, and of 
$\sigma$ for $0 < \sigma \leqslant 3.4$.

\begin{figure}[!ht]
\centering
\begin{minipage}[b]{0.47\textwidth}
\centering
\includegraphics[width=\linewidth]
 {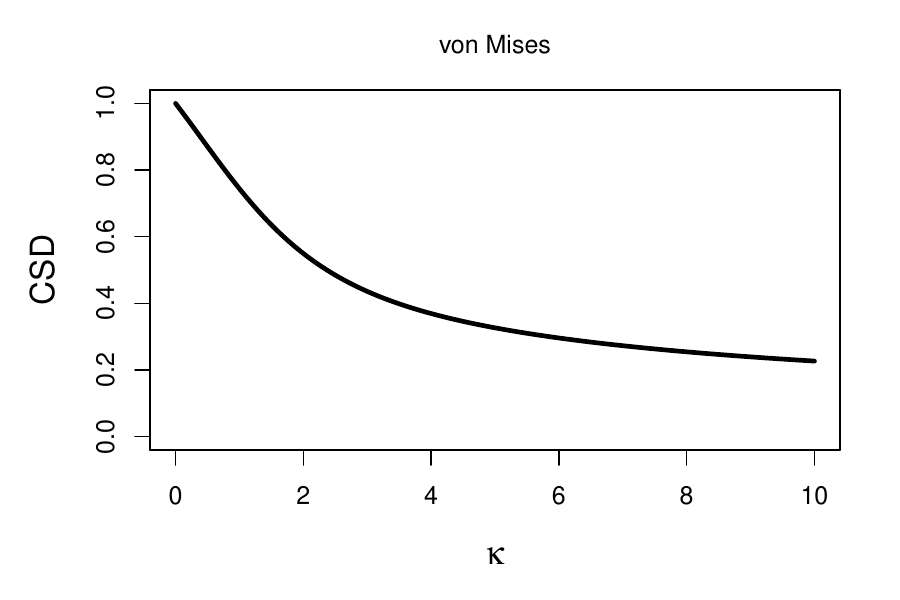}
\end{minipage}%
\hspace{0.02\textwidth}%
\begin{minipage}[b]{0.47\textwidth}
\centering
\includegraphics[width=\linewidth]
{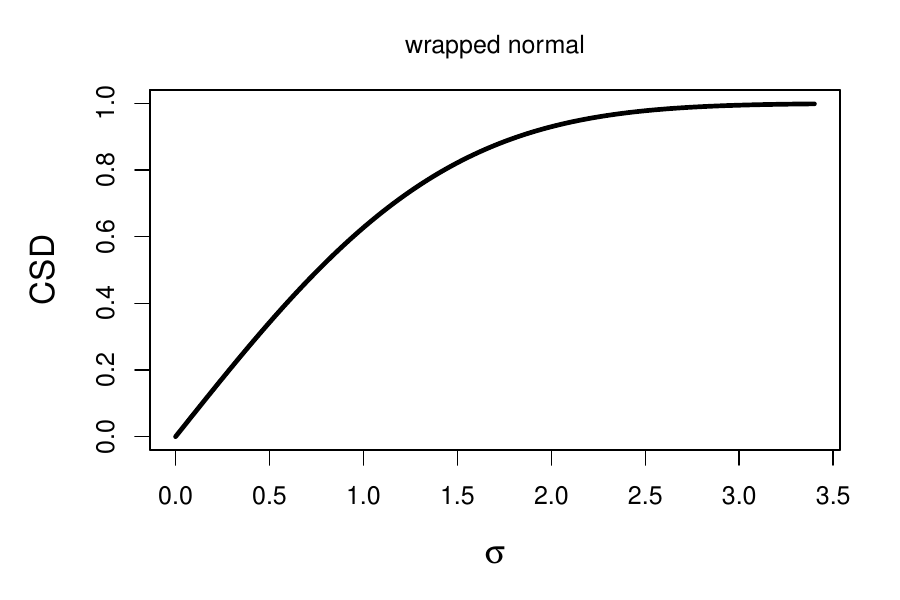}
\end{minipage}
\vspace{-8mm}
\caption{Values of $\CSD$ as a function of the 
parameters $\kappa$ at the von Mises and 
$\sigma$ at the wrapped normal.}
\label{fig:CSD_kappa_Sigma}
\end{figure}

\subsection{Robust consistent estimators of 
\texorpdfstring{$\kappa$}{kappa} and 
\texorpdfstring{$\sigma$}{sigma}}
We now use these relationships  
to construct robust and consistent 
estimators of the model parameters by first 
obtaining a robust estimate of the population
$\CSD$ and then mapping it to the 
corresponding parameter $\kappa$ or $\sigma$ 
through an appropriate transformation.

Let $S_n$ denote a robust circular dispersion 
measure computed from a sample 
$\{\theta_1,...,\theta_n\}$. For each 
circular model, the corresponding population 
dispersion measure $S$ can be expressed as a
function of the underlying $\CSD$ as 
\[
 S=\eta(\CSD).
\]
By inverting this relationship we then obtain a 
robust estimator of the population $\CSD$ by
\[
 \widehat{\CSD}_n=\eta{^{-1}}(S_n).
\]
From this robust estimator of the population
$\CSD$ we then derive consistent estimators of 
the corresponding parameters. 
For the von Mises distribution we obtain an 
estimator of $\kappa$ by numerically inverting the 
function $A(\kappa)=I_1(\kappa)/I_0(\kappa)$, 
yielding the consistency function
$$\widehat{\kappa}=A^{-1}(\widehat{\rho})=A^{-1}
 \Big(1-\frac{\widehat{\CSD}_n^2}{2}\Big).$$
For the wrapped normal we similarly construct
the consistency function for the scale parameter 
$\sigma$ as
$$\hsigma=\sqrt{-2\log(\widehat{\rho})}=
\sqrt{-2\log \Big(1-\frac{\widehat{\CSD}_n^2}{2}
\Big)}\;.$$
In the linear setting the analogous consistency 
function is constant, so it becomes a consistency 
factor.

\subsection{Breakdown values of CMAD, CLMS, and CLTS}

A popular robustness measure of a functional is 
its breakdown value, defined as the highest 
fraction of outliers the functional can cope with. 
For a dispersion/concentration estimator, the 
implosion and explosion breakdown values are defined.
The first is the smallest fraction of contamination 
needed to drive the estimate down to zero, while the 
second is the smallest fraction of contamination 
needed to drive the estimate to its supremum.

A functional $S$ is Fisher consistent for a 
parameter $\psi$ of distributions $F_\psi$
if always $S(F_\psi) = \psi$.

\begin{definition}[Breakdown values of a functional]
\label{def:breakdown}
The implosion and explosion breakdown values of 
a functional $S$ at a circular distribution $F$
are given by
\begin{align*}
\eps^*_{\uimp}(S,F) &:= \inf \{\eps > 0: 
  \inf_H S( (1-\eps)F + \eps H) = 0 \}\\
\eps^*_{\uexp}(S,F) &:= \inf \{\eps > 0:
  \sup_H S( (1-\eps)F + \eps H)
= \sup S \}
\end{align*}
where $H$ ranges over all circular distributions. 
The overall breakdown value is
\begin{equation*}
\eps^*(S, F) : = \min\{\eps^*_{\uimp}(S,F), 
               \eps^*_{\uexp}(S,F) \}.
\end{equation*}
\end{definition}

Let $\psi$ denote a circular dispersion parameter 
and let $\hpsi_{\uS}$ be an estimator defined by 
a strictly monotone consistency function of a 
robust dispersion measure $S$ that is either 
CMAD, CLMS, or CLTS. Then the breakdown values 
of each $\hpsi_{\uS}$ are determined by those of 
the corresponding robust measure $S$, whose 
upper and lower bounds are given in 
Lemmas~\ref{Lemma:inf_sup_vM} 
to~\ref{Lemma:CLTS_Unif}. These breakdown values 
are given below. The proofs are provided in 
Section C of the Supplementary Material. 

\begin{Lemma}
[Upper bound for the implosion breakdown value 
of $\hpsi_{\uCMAD}$] \label{Lemma:UpperboundBP}
Let $F$ be a unimodal and symmetric circular 
distribution with dispersion parameter
$\psi$. We temporarily restrict the $H$ in 
Definition~\ref{def:breakdown} to those
distributions for which 
$G:=(1-\eps)F + \eps H$ has a unique 
circular median. Then
\begin{equation}
\label{eq:upper_bound_imp}
\eps^*_{\uimp}(\hpsi_{\uCMAD},F) < 1 - 
 0.5/[P_F(\mu-\pi/2 < \Theta <  \mu+\pi/2)]\,.    
\end{equation}
Therefore, for both the von Mises and the
wrapped normal distributions, the implosion
breakdown value of $\hpsi_{\uCMAD}$ is less
than 0.5\,.
\end{Lemma}

\begin{theorem}[Overall breakdown value of 
$\hpsi_{\uCMAD}$]\label{Theorem:ExpCMAD}
Let $F$ be a unimodal and symmetric 
circular distribution with dispersion or
concentration parameter $\psi > 0$. We 
temporarily restrict the $H$ in 
Definition~\ref{def:breakdown} to those
distributions for which 
$G:=(1-\eps)F + \eps H$ has a unique 
circular median. Then
\begin{equation*}
\eps^*_{\uexp}(\hpsi_{\uCMAD},F) = 0.5\,.
\end{equation*}
Therefore the overall breakdown value of 
$\hpsi_{\uCMAD}$ satisfies the upper bound
\begin{equation}
\label{eq:upper_bound}
\eps^*(\hpsi_{\uCMAD},F) < 1 - 
 0.5/[P_F(\mu-\pi/2 < \Theta < 
 \mu+\pi/2)] < 0.5\,.    
\end{equation}
\end{theorem}

By Theorem~\ref{Theorem:ExpCMAD}, the 
breakdown value of $\hpsi_{\uCMAD}$ does not 
reach 0.5, the maximum attainable value. 
This contrasts with the MAD for linear 
data, whose breakdown value does reach 0.5. 

On the other hand, although the breakdown value
of CMAD does not attain the value 0.5, the
upper bound can still be relatively high. 
Figure \ref{fig:CMAD_upper_bound} shows the
upper bound of the implosion breakdown value
of CMAD at the von Mises model, as a function 
of the concentration parameter $\kappa$. 
For $\kappa$ as low as two, we attain 0.459. 
The upper bound becomes close to 0.5 when the 
probability of the half circle 
$P_F(\mu-\pi/2 < \Theta <  \mu+\pi/2)$
approaches 1. Note that $\hkappa_{\uCMAD}$
is a monotone decreasing function of
CMAD, so it explodes when CMAD implodes,
but its breakdown value remains the same as
in Theorem~\ref{Theorem:ExpCMAD}. 

\begin{figure}[!ht]
\centering
\includegraphics[width=.6\textwidth]
   {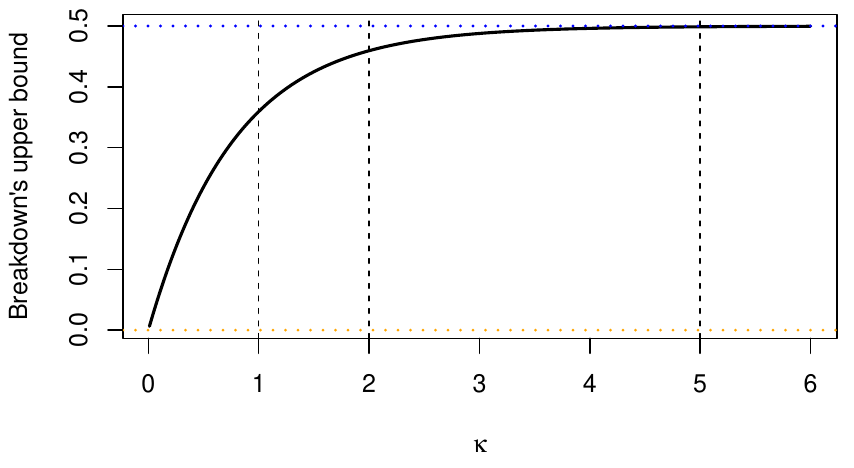}
\caption{Upper bound for the breakdown value 
of $\hkappa_{\uCMAD}$ as a function of the
value of the concentration parameter 
$\kappa$ under a von Mises distribution, from 
Theorem~\ref{Theorem:ExpCMAD}.}
\label{fig:CMAD_upper_bound}
\end{figure}

The estimator of $\psi$ based on $\CLMS$ does
attain the highest possible breakdown value. 

\begin{theorem}
[Breakdown value of $\hpsi_{\uCLMS}$]
\label{Theorem:bdvCLMS}
Let $F$ be a unimodal symmetric distribution 
with dispersion or concentration 
parameter $\psi > 0$, with estimator 
$\hpsi_{\uCLMS}$. Then
\begin{equation*}
\eps^*(\hpsi_{\uCLMS},F) =  0.5.
\end{equation*}
\end{theorem}

This follows directly from 
Lemmas~\ref{Lemma:CLMS=0}
and~\ref{Lemma:CLMS_max}. 

\begin{theorem}
[Breakdown value of $\hpsi_{\uCLTS}$]
\label{Theorem"bdvCLTS}
Let $F$ be a unimodal symmetric distribution 
with dispersion or concentration 
parameter $\psi > 0$, with estimator 
$\hpsi_{\uCLTS}$. Then
\begin{equation*}
\eps^*(\hpsi_{\uCLTS},F) =  0.5.
\end{equation*}
\end{theorem}

This result is analogous to 
Theorem~\ref{Theorem:bdvCLMS}.
In particular, the implosion breakdown value 
is obtained from Lemma~\ref{Lemma:CLTS=0},
and the explosion breakdown value from
Lemma~\ref{Lemma:CLTS_Unif}.

The breakdown properties of the estimators 
$\hpsi_{\uCMAD}$, $\hpsi_{\uCLMS}$ and
$\hpsi_{\uCLTS}$ depend directly on the
implosion and explosion breakdown values of
the robust circular dispersion functionals
$\CMAD$, $\CLMS$, and $\CLTS$, and are 
therefore independent of the assumed 
parametric family.

\subsection{Efficiencies of 
\texorpdfstring{$\hpsi_{\text{CMAD}}$, 
$\hpsi_{\text{CLMS}}$, 
$\hpsi_{\text{CLTS}}$}
{psi hat CMAD, psi hat CLMS, psi hat CLTS}}
\label{sec 6.4}

The efficiency of the proposed robust estimators 
is assessed through their asymptotic relative 
efficiency (ARE) with respect to the maximum 
likelihood estimator (MLE) under the assumed model.
For a parameter of interest $\psi$, the ARE of an 
estimator $\hpsi$ is defined as
\[
\mathrm{ARE}(\hpsi) =
\mathrm{AVar}(\hpsi_{\uMLE})/
\mathrm{AVar}(\hpsi)
\]
where AVar stands for the asymptotic variance. 
These are derived in Section D of the Supplementary 
Material, and can be computed numerically.

For the von Mises distribution, the left panel of 
Figure~\ref{fig:ARE_CMAD_CLMS_CLTS} shows the ARE 
of the estimators $\hkappa_{\uCMAD}$, 
$\hkappa_{\uCLMS}$, and $\hkappa_{\uCLTS}$ over a
grid  of parameters $\kappa$ ranging from 1 to 10.
The ARE of $\hkappa_{\uCMAD}$ and 
$\hkappa_{\uCLMS}$ are the same, and are shown 
as the orange curve, while the ARE of 
$\hkappa_{\uCLTS}$ is shown in purple.

\begin{figure}[!ht]
\centering
\begin{minipage}[b]{0.48\textwidth}
\centering
\includegraphics[width=0.99\linewidth]
  {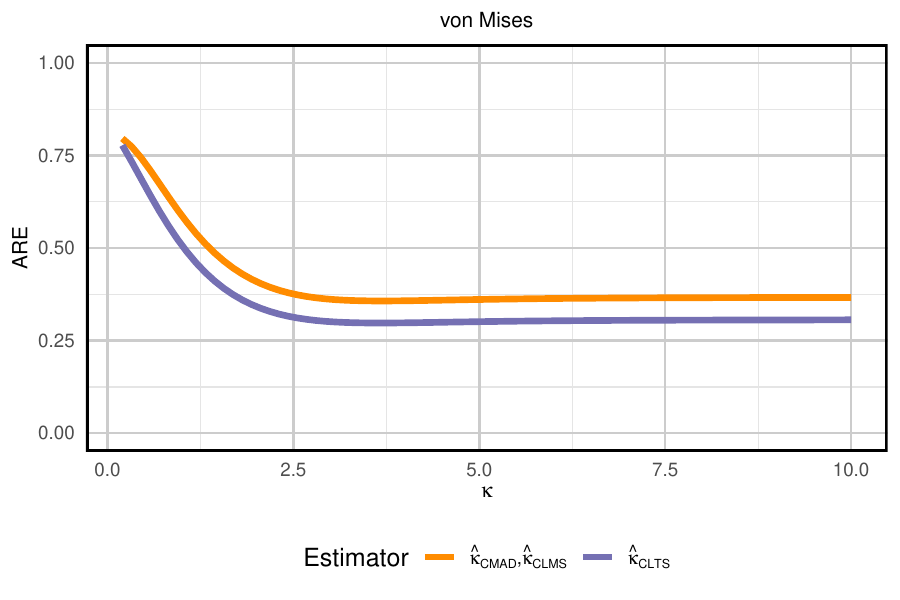}
\end{minipage}%
\hspace{0.04\textwidth}
\begin{minipage}[b]{0.48\textwidth}
\centering
\includegraphics[width=0.99\linewidth]
  {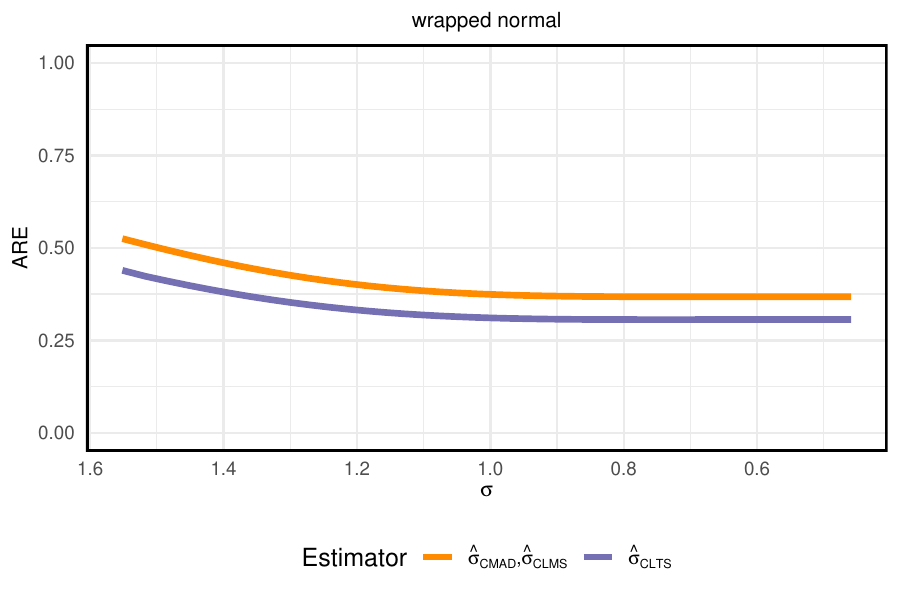}
\end{minipage}
\vspace{-10mm}
\caption{Asymptotic relative efficiency curves 
of the estimators based on $\CMAD$, $\CLMS$, 
and $\CLTS$ as a function of the concentration 
parameter $\kappa$ at the von Mises and as a 
function of the dispersion parameter $\sigma$ 
at the wrapped normal.}
\label{fig:ARE_CMAD_CLMS_CLTS}
\end{figure}

We see that the ARE decreases monotonically as 
$\kappa$ increases. For low $\kappa$, the $\CMAD$ 
and $\CLMS$ based estimators retain over $70\%$ 
efficiency relative to the MLE. As $\kappa$ 
increases, the ARE goes down to around $37\%$.
This horizontal asymptote was expected because 
when $\kappa \rightarrow \infty$, the von Mises 
distribution becomes locally similar to a 
normal distribution on the straight line,
and the asymptotic efficiency of the linear 
MAD and LMS is indeed $36.74\%$.

The behavior of the $\CLTS$ based estimator
of $\kappa$ is similar for low $\kappa$.
When $\kappa$ increases the ARE lies a bit
below that of $\hkappa_{\uCMAD}$ and
$\hkappa_{\uCLMS}$. In the limit for
$\kappa \rightarrow \infty$ it becomes
$30.67\%$ which is the efficiency of the
LTS scale estimator at the normal
distribution on the straight line,
as derived in \cite{Croux1992}.

The right panel of
Figure~\ref{fig:ARE_CMAD_CLMS_CLTS} shows the 
ARE of the estimators $\hsigma_{\uCMAD}$, 
$\hsigma_{\uCLMS}$, and $\hsigma_{\uCLTS}$ at 
the wrapped normal distribution. The parameter
$\sigma$ ranges from 0.45 to 1.55\,. Note that it
is plotted from right to left on the horizontal
axis to facilitate comparisons with the plot for
the von Mises distribution in the left panel,
since a lower scale $\sigma$ has a similar
effect as a higher concentration $\kappa$.
The orange curve is the ARE of $\hsigma_{\uCMAD}$ 
and $\hsigma_{\uCLMS}$, and the purple curve is
the ARE of $\hsigma_{\uCLTS}$\,. Again the
CLTS-based estimator is a bit less efficient 
than the CLMS-based one, and for
$\hsigma \rightarrow 0$ we get the same limits
as before, $36.74\%$ for $\hsigma_{\uCLMS}$
and $30.67\%$ for $\hsigma_{\uCLTS}$\,.

Across both distributional models, we can 
conclude that the efficiency of the CLMS-based
estimators slightly outperforms those based
on CLTS.

\subsection{Robustness of 
\texorpdfstring{$\hpsi_{\text{CMAD}}$, 
$\hpsi_{\text{CLMS}}$, $\hpsi_{\text{CLTS}}$}
{psi hat CMAD, psi hat CLMS, psi hat CLTS}: 
an empirical study}

So far we have studied the influence function,
relative bias, and breakdown value. These are
all based on functionals, that are population
quantities. But of course we are also
interested in the behavior on finite samples.
For this purpose we performed a small simulation 
study about the impact of $10\%$ and $20\%$ of 
contamination.

\begin{figure}[!hb]
\centering
\begin{minipage}[b]{0.48\textwidth}
\centering
{Mean shift outlier model}
\includegraphics[width=0.95\linewidth]
{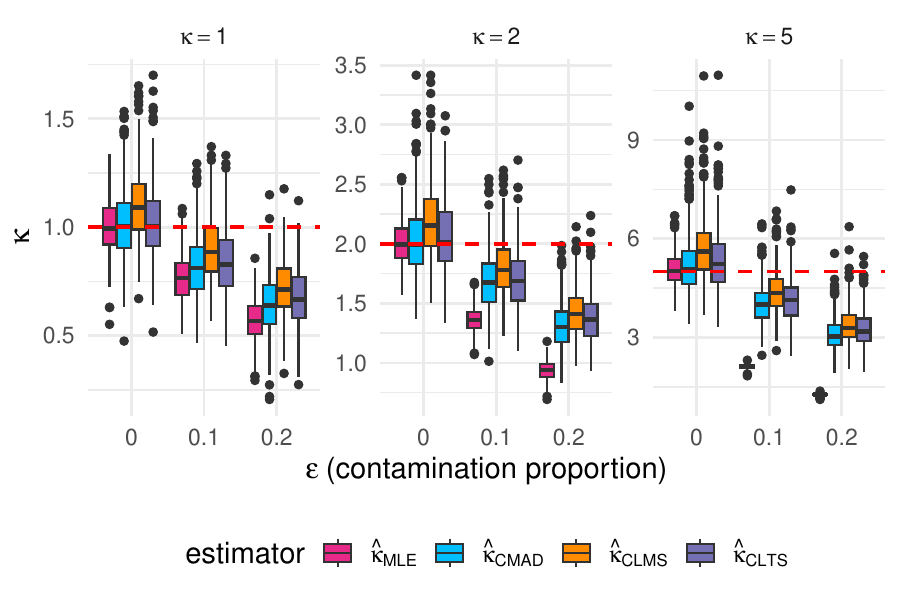}
\end{minipage}%
\hspace{0.04\textwidth}
\begin{minipage}[b]{0.48\textwidth}
\centering
{Point mass contamination}
\includegraphics[width=0.95\linewidth]
{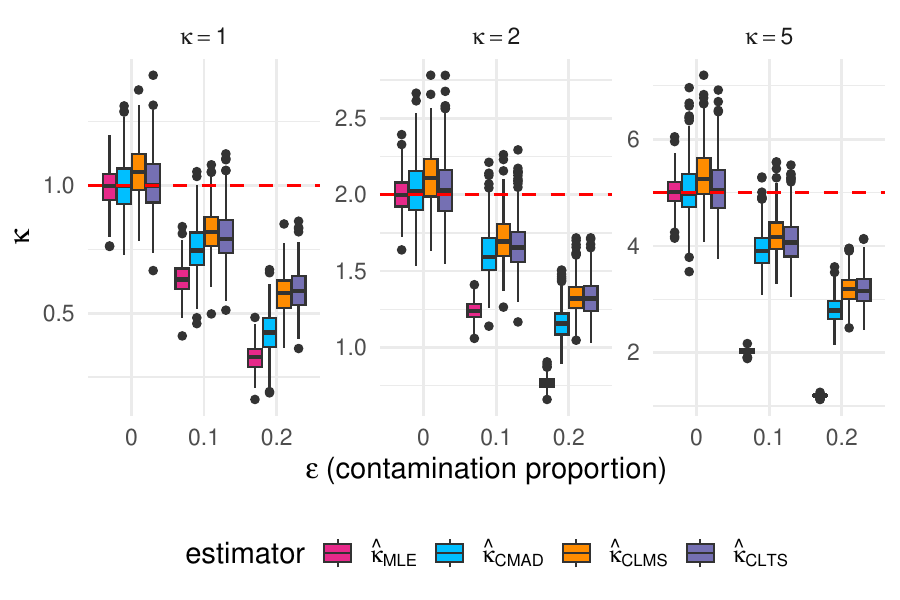}
\end{minipage}
\caption{Simulation results for the von Mises 
distribution with $\kappa \in \{1,2,5\}$, 
with a fraction $\eps$ of antipodal 
contamination. The boxplots of the estimates 
$\hkappa_{\uMLE}$, $\hkappa_{\uCMAD}$, 
$\hkappa_{\uCLMS}$, and $\hkappa_{\uCLTS}$ 
are shown together with the dashed red line
at the true $\kappa$.}
\label{fig:robustness_empirical_VM}
\end{figure}

We start with the von Mises distribution, 
where we compare the effect of 
contamination on the maximum likelihood 
estimator $\hkappa_{\uMLE}$
with that on $\hkappa_{\uCMAD}$, 
$\hkappa_{\uCLMS}$, and $\hkappa_{\uCLTS}$.
We generated 200 points from a von Mises
distribution $F$ with $\mu(F)=0$ and true
concentration $\kappa$ equal to 1, 2, and 5.
Apart from this clean dataset, we also 
replaced 20 of its points (10\%) by points 
from a von Mises distribution $H$ with the 
same $\kappa$ but with antipodal center,
that is, $\mu(H) = \pi$. This 
mean shift outlier model is the
worst contamination for location, see
e.g.\ \cite{Demni2021}. The same scenario 
was repeated for 20\% of contamination.
The experiments were replicated 500 times,
with the results shown in the left panel of 
Figure~\ref{fig:robustness_empirical_VM}.
It contains a boxplot of $\hkappa$ for
each combination of $\kappa$, $\eps$, 
and estimator. For $\eps = 0$ (clean
data) the medians of the four estimators
are similar and lie close to the dashed 
red horizontal line at the true $\kappa$,
except for $\hkappa_{\uCLMS}$ that has
a small upward bias.
The height of each box reflects the
accuracy (efficiency) of the estimators.
As expected, for clean data the MLE is the 
most accurate, but the other estimators 
still perform quite well. The situation
changes for $\eps = 0.1$, where the MLE
is the most affected, followed by CMAD and
then CLTS and CLMS. For $\eps = 0.2$ the
effect is magnified, with the box of MLE
far away and often small.

The right hand panel of 
Figure~\ref{fig:robustness_empirical_VM}
is the result of a simulation with the same
clean data, but now the 10\% and 20\% of 
contaminated data are all put in a single 
point, the antipodal point $\pi$. This is
point mass contamination. We see that the 
effects are qualitatively similar to the 
left panel.

The simulation was repeated for the wrapped
normal distribution, this time with the
dispersion parameter $\sigma$ taking the
values 0.5, 1, and 1.2\,. The contamination
was again of the mean shift outlier type,
and point contamination. The results in
Figure~\ref{fig:robustness_empirical_WN}
look quite similar to those for the von
Mises distribution, taking into account
that the higher estimated dispersions 
$\hsigma$ under contamination play the 
same role as the lower concentrations 
$\hkappa$ before. 

\begin{figure}[!ht]
\centering
\begin{minipage}[b]{0.48\textwidth}
\centering
 {Mean shift outlier model}
\includegraphics[width=0.95\linewidth]
{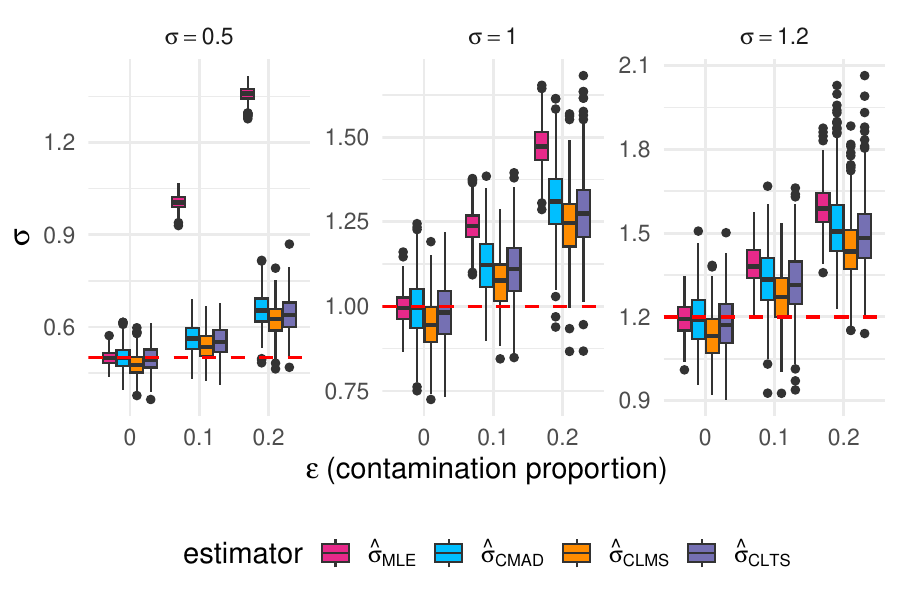}
\end{minipage}%
\hspace{0.04\textwidth}
\begin{minipage}[b]{0.48\textwidth}
\centering
 {Point mass contamination}
\includegraphics[width=0.95\linewidth]
{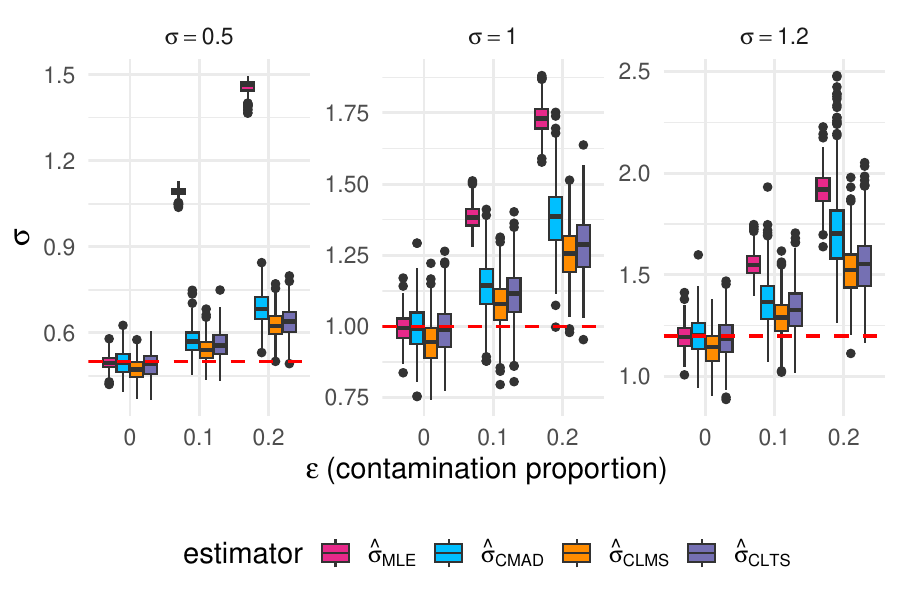}
\end{minipage}
\caption{Simulation for the wrapped
normal distribution with 
$\sigma \in \{0.5,1,1.2\}$ 
and a fraction $\eps$ of antipodal 
contamination. The boxplots of 
$\hsigma_{\uMLE}$, $\hsigma_{\uCMAD}$, 
$\hsigma_{\uCLMS}$, and $\hsigma_{\uCLTS}$ 
are shown together with the dashed red line
at the true $\sigma$.}
\label{fig:robustness_empirical_WN}
\end{figure}

Overall, the simulation results indicate 
that the MLE estimator is the most affected
by contamination, in both the von Mises and 
the wrapped normal models. The three robust
estimators are much less affected. Between
them, CLMS appears to be the most robust.
Taken together with the results of the 
relative bias and the ARE, CLMS appears to 
be the best performing among the three 
robust methods.

\section{Robust anomaly detection}
\label{sect:violin}

In this section we will illustrate how robust
estimation can be used for outlier detection 
on the circle.

\subsection{A robust circular anomaly detection rule}

There is a growing awareness of the effects of 
anomalous observations on data analysis, and the 
need for outlier detection methods for circular data.
Some authors suggested techniques for circular outlier 
detection based on the deletion approach, see e.g.\ 
\cite{Collett1980}, \cite{Jammalamadaka2001}, 
and \cite{Mardia1975}. However, such procedures
often suffer from the masking effect, meaning that
an anomalous point cannot be detected in the presence 
of another outlier. Some of these techniques are only
suitable for small sample sizes, or they assume that 
the concentration parameter is known beforehand.

Other approaches have also been proposed. However, 
some of them are only capable of detecting a single 
outlier \citep{Abuzaid2012}, while 
others use simulations under a specific data model 
to get cutoff values \citep{Mahmood2017}, or 
require that outliers be very well separated from 
the remainder of the data \citep{Mohamed2016}.

Also methodologies for outlier detection in toroidal 
and spherical data have been introduced by
\cite{Abuzaid2020}, \cite{Greco2021}, and 
\cite{Sau2018}. 

Some alternative methods based on robust measures 
have been developed. Both weighted likelihood and 
minimum disparity techniques have been adapted to 
the circular setting \citep{Agostinelli2007}. It
turns out that the performance of these techniques
depends strongly on the choice of certain 
parameters. In a related field, a robust outlier 
detection tool has been considered in circular 
regression \citep{Rana2016}. 

Robust anomaly detection techniques assume that 
most of the data follow some parametric model.
The parameters are then estimated robustly, and 
afterward the outliers are flagged by their large 
distance from the fitted 
model \citep{Rousseeuw2018}.

\cite{Demni2023} applied this approach to 
circular data by assuming a von Mises distribution.
However, the robust estimator of $\kappa$ used
was rather simple and presents high variability 
when the data are moderately concentrated. 

In line with this approach, and making use of the 
results described above, we propose instead to 
use the \mbox{CLMS-based} estimator 
$\hkappa_{\uCLMS}$ for the von Mises distribution, 
and the estimator $\hsigma_{\uCLMS}$ for the 
wrapped normal. We also want the outlier 
detection rule to be invariant to rotation.

The proposed method first estimates the location 
robustly by the Fr\'echet median $\tilde{\mu}$, 
and then computes the shortest arc distances 
$d(\theta_i, \tilde{\mu})$ of each circular 
observation $\theta_i$ to $\tilde{\mu}$. 

These distances are then compared with a cutoff 
value $c_\alpha(\psi)$, where $\alpha>0$ is the 
expected fraction of flagged points when the data 
are clean, and $\psi$ is the dispersion or 
concentration parameter of the underlying 
distribution. We estimate this parameter by the 
CLMS-based robust estimator such as
$\hkappa_{\uCLMS}$ or $\hsigma_{\uCLMS}$.

When the underlying circular distribution is 
symmetric about its location parameter, the 
cutoff value $c_\alpha(\psi)$ is defined as 
the $1-\alpha/2$ quantile of the underlying 
distribution with center at $\mu=0$. It
thus corresponds to the numerical solution 
of the equation
\begin{equation}\label{eq:cutoff}
 \int_{0}^{c_{\alpha}(\psi)} f(\theta, \psi)
 d\theta = 1-\alpha/2\;,
\end{equation}
where $f(\theta,\psi)$ is the density 
of the distribution with $\mu = 0$.
Next, we flag a circular data point 
$\theta_i$ whenever 
\begin{equation}
  d(\theta_i, \tilde{\mu}) 
  > c_{\alpha}(\hpsi_{\uCLMS})\,.
\label{eq_outdet}
\end{equation}

\subsection{Visualizing anomalies: the 
  circular violin plot}

In addition to the outlier detection rule,
we also provide a companion graphical tool,
that we call a {\em circular violin plot}.
Its goal is to enable further data analytic 
insights.

Violin plots on the straight line were introduced 
by \cite{Hintze1998}. They were designed to 
overcome some limitations of traditional boxplots.
When Tukey introduced boxplots little computing
power was available, so their simplicity was
an advantage, but that also limited the amount
of information they could visualize. Violin
plots contain a density estimate that can
reveal multimodality. For an overview of linear
violin plots and their uses see 
\cite{Aberasturi2014}. The violin plot has been 
widely adopted in various disciplines including 
biostatistics \citep{Mekkes2014} and psychology 
\citep{Tanious2022}. Several extensions of 
violin plots and related graphical tools have 
been proposed \citep{bixplot}.

A violin plot for circular data has to respect 
the geometry of the circle. The estimated density 
is mounted on the circumference of the unit 
circle. This density curve is then mirrored 
radially on the inside of the circle. The area
between the density curves forms a smooth ring 
whose thickness reflects the local data 
concentration. The density is estimated 
nonparametrically, using a von Mises kernel.
The density is clipped at the cutoff threshold 
of the outlier detection rule, and data points
lying outside the violin are flagged as outliers.

Inside the violin, elements of the circular
boxplot \citep{Buttarazzi2018} are integrated.
The box represents the interquartile range 
with 25\% of the data on either side of the
median, highlighting the central 50\% of the 
data. The median is indicated by a little white 
line. The median and the box visualize robust 
summaries of location and variability.  
We will illustrate the circular violin plot
on real examples in the next section.

\section{Real data analysis} \label{sect:data}

Three real data examples are considered. The results
shown are based on fitting the data by von Mises 
distributions, but using the wrapped normal 
distribution yields essentially the same results.

\subsection{Northern cricket frogs data}
\cite{Ferguson1967} collected data on the homing 
ability of Northern cricket frogs. After 30 hours
in a dark environmental chamber, 14 frogs were 
released and their directions recorded. 
Results are reported in radians, with zero 
radians corresponding to East, and angles 
measured counterclockwise. The data are plotted 
in the left panel of 
Figure \ref{fig:outlierdetectionfrogs}. 
Their circular sample median 
$\hmu=2.39$ radians is indicated by a 
black arrow. The CLMS-based robust estimate of 
the concentration parameter is 
$\hkappa_{\uCLMS}=3.88$. From 
formula~\eqref{eq:cutoff} we obtain the cutoff 
value $c_{0.01}(\hkappa_{\uCLMS})=1.34$ radians.

The circular violin plot in the right panel is 
constructed by estimating the circular density 
using the robust estimate $\hkappa_{\uCLMS}$. 
The density is rescaled and displayed radially 
around the unit circle, yielding a symmetric 
violin shape. One observation is flagged as an 
outlier. This is in line with the literature: 
\cite{Collett1980} found the same point 
(5.515 radians) to be a circular outlier, by 
means of a deletion approach.

\begin{figure}[!ht]
\centering
\vspace{-3mm}
\includegraphics[height=5.5cm, trim=0 80 0 0, clip]
{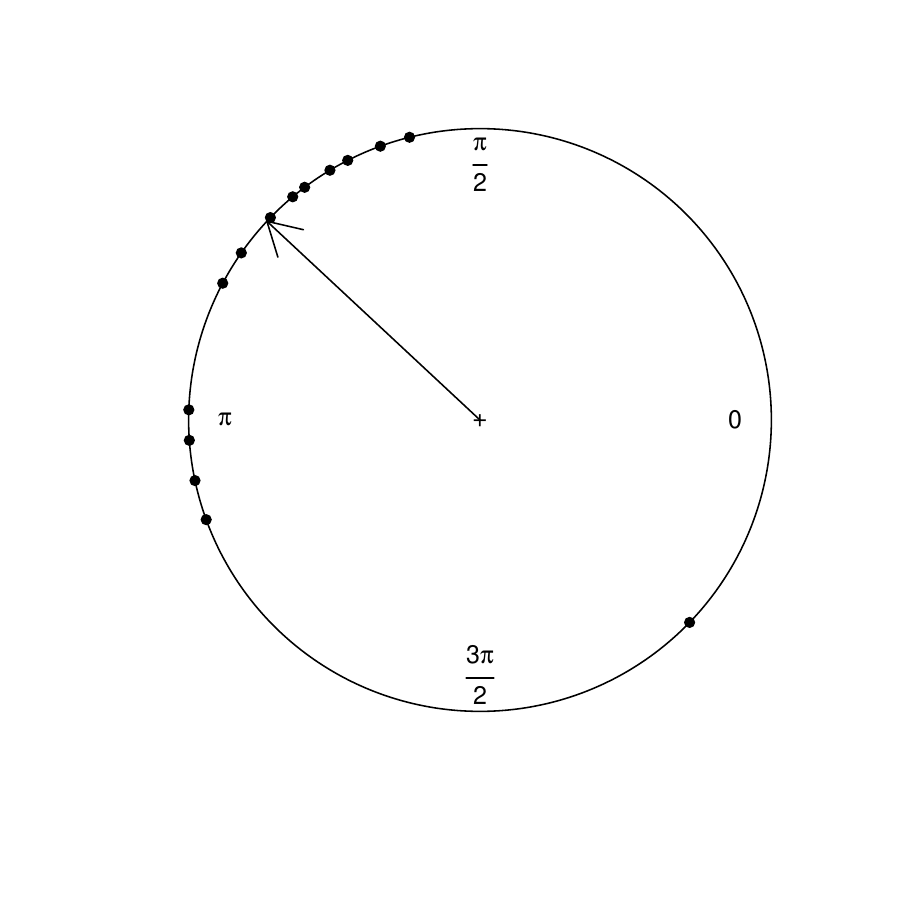}%
\hspace{0.05\textwidth}
\includegraphics[width=0.49\textwidth]
{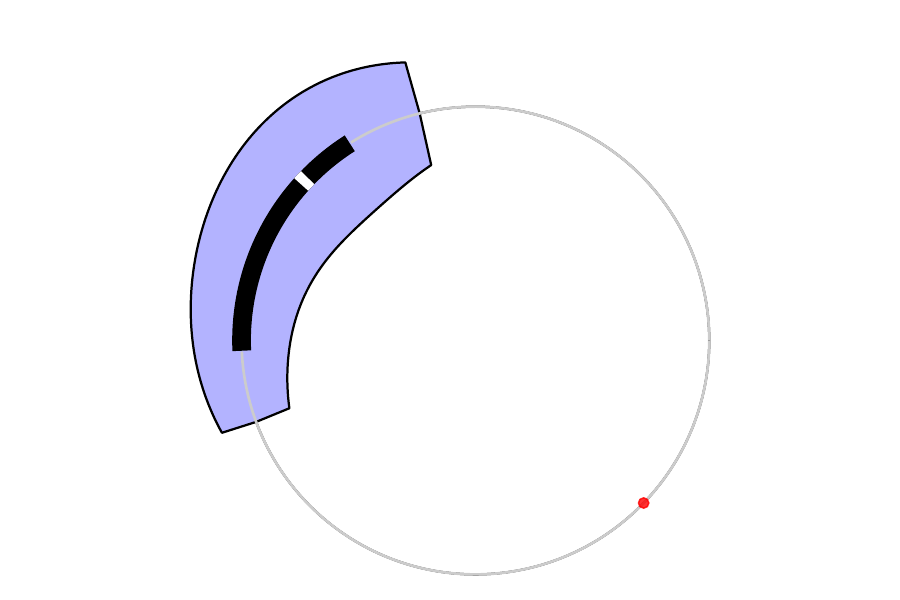}
\caption{(Left) Plot of the Northern cricket frogs 
data with their sample circular median direction 
(black arrow), and (right) their circular violin 
plot.}
\label{fig:outlierdetectionfrogs}
\end{figure}

At first sight one might think that the parts of
the violin inside and outside the circle are not
mirrored exactly, but they are. The perception
is a consequence of the circular geometry 
itself, with its curvature and the closure 
of the circular domain. In a linear violin plot, 
the density is reflected across a straight axis,
so equal distances on either side of the axis 
correspond to equal areas. On the circle, the 
mirroring is performed along radial directions,
which does not translate into equal areas 
or visually symmetric shapes.
As a result, one might perceive the mirrored 
density as a single continuous band of varying 
thickness rather than as two symmetric halves. 
This effect is inherent in circular geometry 
and is independent of the specific density 
estimator or scaling choice.

\subsection{Sardinian sea stars data}
We saw the Sardinian sea stars data in
Figure~\ref{fig:SeaStars}. They are the 
resultant directions of 22 sea stars, eleven 
days after being displaced from their 
natural habitat \citep{Fisher1995}. 
They are displayed in the left panel of
Figure~\ref{fig:outlierdetectionSeaStars}.
The direction labeled as 0 radians (shown
at the top) is the `correct' direction
of sea stars traveling toward the shore.

In this dataset, two points are far away from 
the majority of the data: one sea star 
(at 5.201 radians counterclockwise)
moved along a longer route towards the shore, 
while the second (at 2.565 radians) navigated 
away from it. The question is whether these 
two points can be considered outliers.

If the single-value deletion technique in 
\cite{Collett1980} is adopted, neither of them 
is flagged as outlying, due to the masking 
effect.

\cite{Kato2016} also applied a single-value
deletion technique. They compared robust 
estimates of the location and the concentration 
parameters of the von Mises distribution to the 
corresponding MLE estimates 
on subsamples where one point was omitted. 
They concluded that the point
at 2.565 appears to be an outlier, while the 
outlyingness of the point at 5.201 was much 
more difficult to judge. This is 
understandable: if the furthest point is kept 
in the dataset, the omission of the second
furthest does not have a large effect on the 
MLE.

A sequential deletion approach (i.e., 
deleting the most outlying observation and 
repeating the procedure) 
also has limitations in these data. According 
to the analysis in \cite{Fisher1995}, a von 
Mises Q-Q plot visually suggested that the 
point at 2.565 radians is discordant. But 
after deleting this point the von Mises Q-Q 
plot on the remaining data looked okay. The
test for discordancy from the von Mises 
distribution did not reject the null 
hypothesis of no discordancy, so the point 
at 5.201 radians did not get flagged.

\begin{figure}[!ht]
\centering
\vspace{-5mm}
\includegraphics[height=5.5cm, trim=0 80 0 0, clip]
{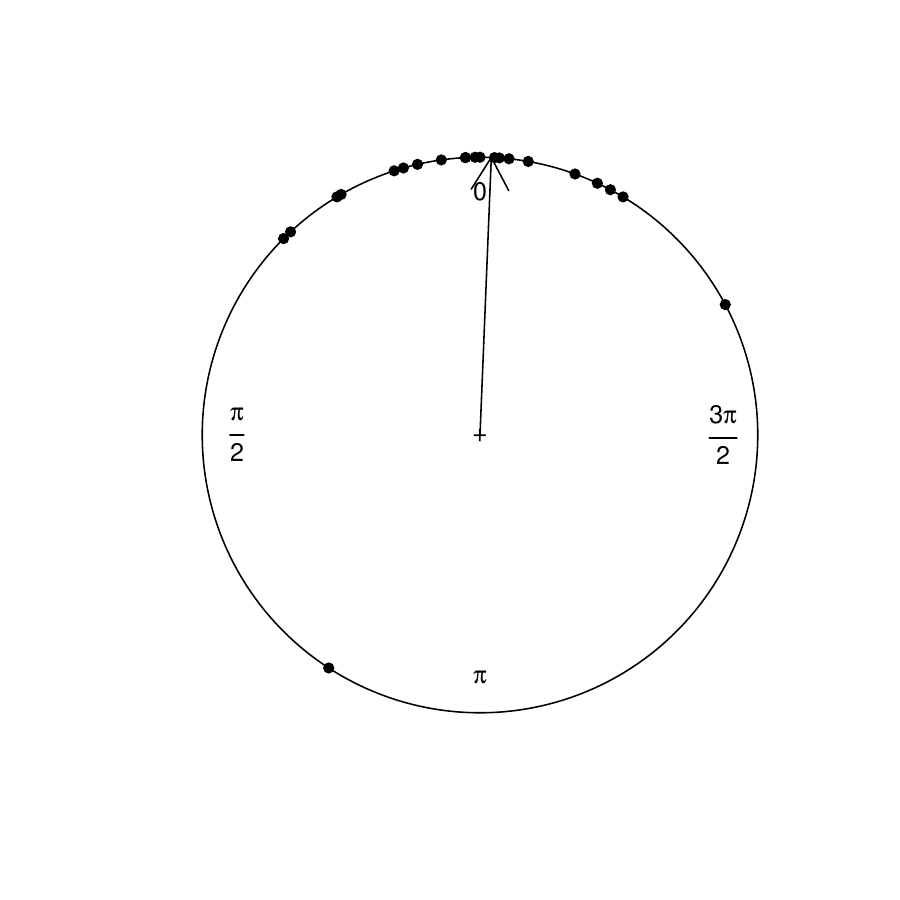}%
\hspace{0.05\textwidth}
\includegraphics[width=0.49\textwidth]
{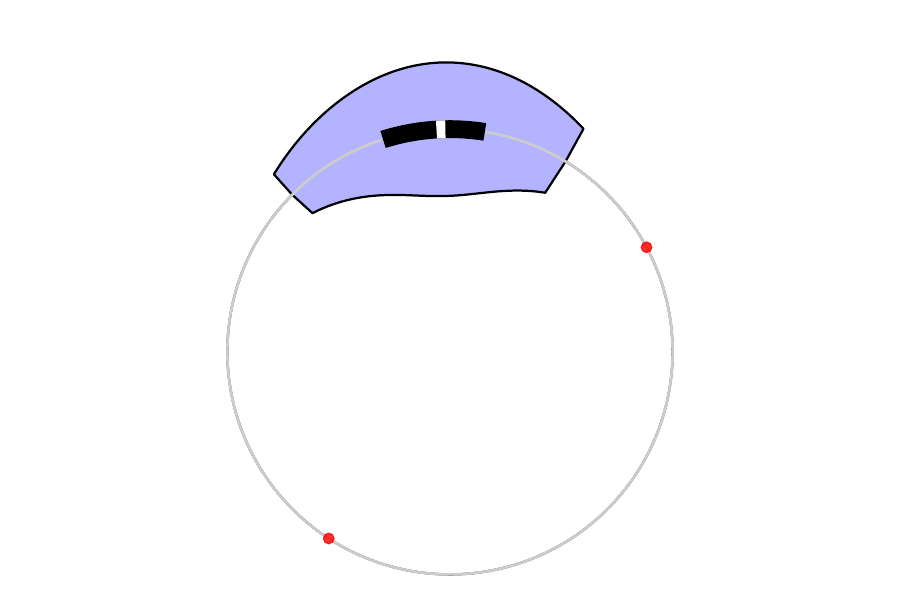}
\caption{(Left) Plot of Sardinian sea stars data
with circular median direction (black arrow), and  
(right) circular violin plot.}
\label{fig:outlierdetectionSeaStars}   
\end{figure}

Our proposed outlier detection rule tells 
a different story. From the 
$\CLMS$-based estimate $\hkappa_{\uCLMS}=7.92$
we obtain the threshold 
$c_{0.01}(\hkappa_{\uCLMS})=0.86$ radians. The
resulting circular violin plot is in the right 
panel of Figure~\ref{fig:outlierdetectionSeaStars}. 
It clearly flags both far-out points as outliers. 
This is because $\hkappa_{\uCLMS}$ uses only $50\%$ 
of the data points to estimate $\kappa$, so neither
of the two far-out points enter that calculation. 
On the other hand, if the $\MLE$ location 
$\hmu_{\uMLE}$ and concentration 
$\hkappa_{\uMLE}$ had been used in 
formula~\eqref{eq_outdet}, the corresponding 
cutoff value would have been 
$c_{0.01}(\hkappa_{\uMLE})=1.52$, and then only 
the furthest point would have been flagged.

\subsection{Drosophila fly larva data}

\cite{Pewsey2012} recorded changes in movement of 
a Drosophila fly larva, measured once per second
over a period of three minutes. The left panel of
Figure~\ref{fig:outlierdetectionLarva} displays 
all 180 points, with the reference direction of 0 
radians shown at the top. The sample circular median 
is $\tilde{\mu}=6.26$ radians (near the top) and 
the $\CLMS$-based estimate of concentration is 
$\hkappa_{\uCLMS}=15.58$. Assuming the data comes 
from a von Mises distribution, the corresponding 
cutoff value is $c_{0.01}(\hkappa_{\uCLMS}) = 0.60$. 

\begin{figure}[!ht]
\centering
\vspace{-5mm}
\includegraphics[height=5.5cm, trim=0 70 0 0, clip]
{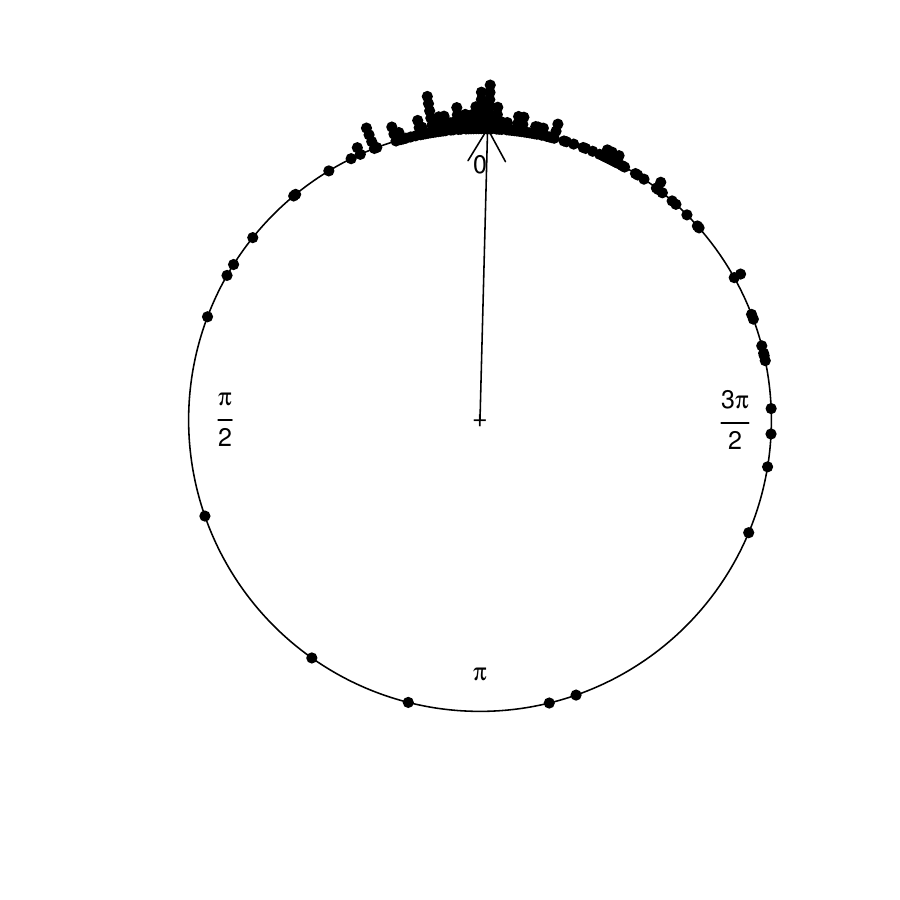}%
\hspace{0.05\textwidth}
\includegraphics[width=0.49\textwidth]
{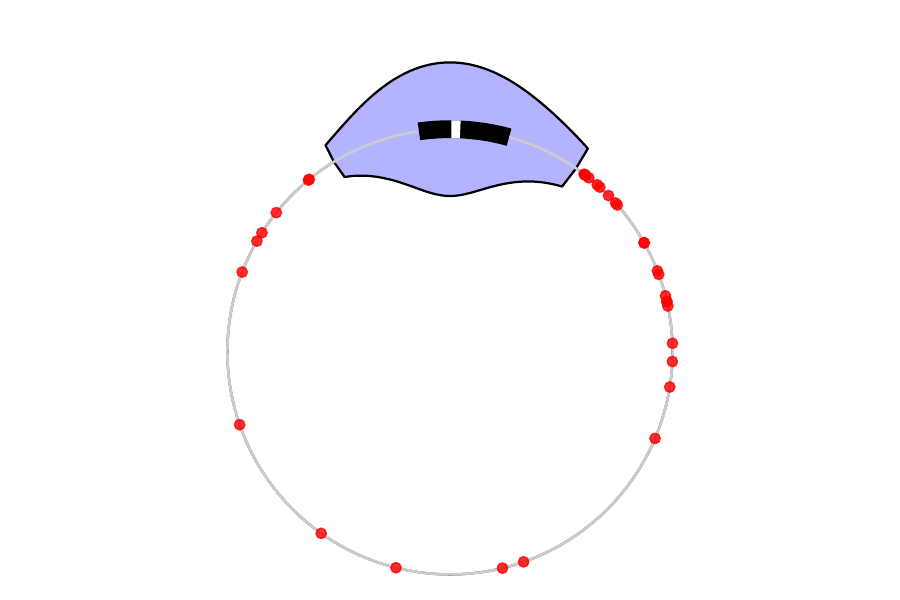}
\caption{(Left) plot of the Drosophila fly larva data
with its circular median (black arrow), and 
(right) its circular violin plot.}
\label{fig:outlierdetectionLarva}   
\end{figure}

In the circular violin plot in the right panel,
many points are flagged as outliers. Furthermore, 
there is no clear gap between the violin and the 
nearby `outliers'. This suggests that the data, 
rather than being generated 
by a von Mises distribution, come from a 
distribution with heavier tails. The results of
\cite{Pewsey2012} indeed suggest that the larva 
data can be better modeled by an Inverse 
Batschelet distribution. 
We conclude that the circular violin plot can be
of some help in questioning the distributional
model underlying the observed data.

\section{Conclusions} \label{sect:concl}

When analyzing circular data it is important to
be able to estimate its dispersion. Classical
dispersion measures such as the circular standard
deviation can be unduly affected by outliers.
Here we investigated three robust estimators of
dispersion, CMAD, CLMS, and CLTS, that are 
extensions of popular dispersion measures for 
linear data. Of these only CMAD was used before, 
but not yet analyzed by the tools of robust 
statistics. In this paper we did, by deriving
the influence function, relative bias, breakdown
value, and asymptotic efficiency of all three
methods, and we performed a simulation study
with clean and contaminated data. 
Comparing all of these results we concluded that
CLMS performed the best.

We have also constructed an outlier detection
rule based on CLMS, that is less vulnerable to
masking than earlier approaches. It is
accompanied by a new graphical display called
a circular violin plot. These tools enabled
us to analyze three real datasets with varying
numbers of potential outliers.

While the present work focused on circular 
data, extensions to other manifolds such as 
the sphere or the torus would be of 
considerable interest.\\

\noindent \textbf{Software}
A zip file with the \textsf{R} code and an 
example script reproducing the examples is
available at 
\url{https://wis.kuleuven.be/statdatascience/code/circular_r_code.zip}\,.\\

\noindent \textbf{Funding}
The research of H. Demni was supported by the 
Flanders Research Foundation (FWO) under a Grant
for a scientific stay in Flanders V500325N, and 
was also carried out within the project 
ECS 0000024 “Ecosistema dell’innovazione - Rome 
Technopole” financed by EU in NextGenerationEU 
plan through MUR Decree n. 1051 23.06.2022 PNRR
Missione 4 Componente 2 Investimento 1.5 - 
CUP H33C2200042000. The work of G. Porzio was 
funded by the project ECS 0000024 “Ecosistema 
dell’innovazione - Rome Technopole” financed by 
EU in NextGenerationEU plan through MUR Decree 
n. 1051 23.06.2022 PNRR Missione 4 Componente 2 
Investimento 1.5 - CUP H33C2200042000.\\

\clearpage

\pagenumbering{arabic}

\appendix

\renewcommand{\thefigure}{S\arabic{figure}}
\setcounter{figure}{0}

\begin{center}

\LARGE{\bf Supplementary Material}\\
\end{center}


\section*{A: Robust dispersion measures}

\noindent \textbf{Proof of Lemma~\ref{Lemma:CMAD=0}}\\ 

To prove Lemma~\ref{Lemma:CMAD=0}, we first state 
an intermediate Lemma.

\begin{Lemma}
Let $\Theta$ be a circular random variable, and 
$\theta^* \in [-\pi, \pi)$ be such that 
$P(\Theta = \theta^*) \geqslant 0.5$, then the 
circular median set $M_\theta = \{\theta^*\}$. 
That is, $\theta^*$ is the circular median of 
$\Theta$.
\end{Lemma}

\begin{proof}
Consider a value $\delta \in [-\pi, \pi)$ such 
that $\delta > \theta^*$. We prove that
\begin{equation}\label{eqn:phi*}
  \mathbb{E}(d(\Theta, \delta)) \geqslant
  \mathbb{E}(d(\Theta, \theta^*)) \hspace{12pt}
  \forall \hspace{2pt} \delta > \phi^*.   
\end{equation}
Since the case $\delta < \theta^*$ follows by 
symmetry, this implies $M_\theta = \theta^*$, 
by definition of the circular median.

Let $\gamma \in [-\pi, \pi)$. By the properties 
of arc distances, we have that
\begin{equation*}
  d(\gamma, \delta) -  d(\gamma, \theta^*) = 
  \begin{cases}
    d(\delta, \theta^*) & \text{if } 
    \gamma \leqslant \theta^*\\
    > -d(\delta, \theta^*) & \text{if } 
   \gamma > \theta^*
  \end{cases}
  \hspace{18pt}
  \forall \hspace{2pt} \delta > \theta^* \, .   
\end{equation*}
Hence, by considering the difference 
$\mathbb{E}(d(\Theta, \delta)) - 
\mathbb{E}(d(\Theta, \theta^*))$  
and conditioning, we have
$\mathbb{E}(d(\Theta, \delta) - d(\Theta, \phi^*))
\geqslant d(\delta, \theta^*) 
P(\Theta \leqslant \theta^*) - d(\delta, \theta^*)
P(\Theta > \theta^*) =
d(\delta, \theta^*) (P(\Theta \leqslant \theta^*)
- P(\Theta > \theta^*)) \geqslant 0$,
since $P(\Theta = \theta^*) \geqslant 0.5$ by 
assumption.
\end{proof}  
Note that, by the same argument, if 
$P(\Theta = \theta^*) > 0.5$, then 
$\mathbb{E}(d(\Theta, \delta)) > 
\mathbb{E}(d(\Theta, \theta^*)) \hspace{2pt}
\forall \hspace{2pt} \delta > \theta^*$, and 
$\theta^*$ is the unique minimizer. Hence, 
$\Theta \in {\mathcal M}_1$.\\

The proof of Lemma~\ref{Lemma:CMAD=0} follows, 
equation by equation.

\begin{proof}
If $\exists \hspace{2pt} \theta : 
P(\Theta = \theta) > 0.5$, then 
$M_\theta=\theta$ and $\Theta \in {\mathcal M}_1$ 
according to the above lemma. We thus have 
$P(d(\Theta,\theta) = 0) > 0.5$, and 
$P(d(\Theta,\theta) \geqslant 0) > 0.5$. Hence, 
by definition of the linear median, $\CMAD(F)=0$.

If $ \exists \hspace{2pt} \theta : 
P(\Theta = \theta) = 0.5$, then $\theta$ belongs 
to the set $M_\theta$ according to the above 
lemma. If that set contains more than one
element, then the $\CMAD$ is undefined. If it is 
a singleton instead, then the distribution of 
the distances $d(\Theta, \theta)$ can have median 
equal to zero, or zero can be the extreme of a 
(linear) median set. In the first case (for 
instance when $F = 0.5 \Delta_\theta + 0.5 G$, 
with  $\Delta_\theta$ a point mass distribution 
located in $\theta$ and 
$G\sim \text{Uniform}(-\pi,\pi)$), $\CMAD(F)=0$. 
In the second case (e.g., if 
$F= 0.2 \Delta_{\theta-\pi/8} + 
0.5 \Delta_{\theta} + 0.3\Delta_{\theta+\pi/6}$), 
$\CMAD(F) > 0$ being $\CMAD$ the midpoint of an 
interval of positive length that has an extreme 
in zero (this happens because the cdf of the 
distribution of the distances from $\theta$ is 
a function that takes value 0.5 from zero up to 
a certain point). Finally, because the distance 
distribution is bounded by $\pi$, we also have 
$\CMAD(F) < \pi/2$. The equality cannot be 
reached because it would necessarily to require 
that $F = 0.5\Delta_\theta + 
0.5\Delta_{\theta \pm \pi}$, which is indeed a 
case of $\Theta \notin {\mathcal M}_1$.

If $\forall \hspace{2pt}\theta \hspace{2pt} 
P(\Theta = \theta) < 0.5$, then the distribution 
of the distances from $\theta$ cannot have 
median equal to zero and therefore 
$\CMAD(F) > 0$. Furthermore, that median cannot 
reach the value of $\pi$. Because the circle is
bounded, the distribution of the shortest arc 
distance has its maximum achievable value equal 
to $\pi$, and hence $\CMAD(F) \leqslant \pi$. 
However, to get $\CMAD=\pi$, the (linear) 
median should be equal to the maximum achievable 
value of the distance distribution. This can 
occur only if at least half of the probability 
is in the upper tail of the distance 
distribution, which in turn implies that (at 
least) half of the probability is pointwise 
located somewhere on the circle. However, by 
the above Lemma, $\CMAD$ would be either equal
to 0, or to a set that includes 0, not to 
$\pi$. Consequently, $\CMAD(F) < \pi, 
\hspace{2pt} \forall \; \Theta \in 
{\mathcal M}_1$. 
\end{proof}

\noindent \textbf{Proof of 
Theorem~\ref{Lemma:CMAD_Ubound}}
To prove the theorem, we show there exists a 
family of distributions for the supremum of 
$\CMAD$ is $\pi$. To start, consider the 
distributions
\begin{equation*}\label{eq:CMADsup_proof}
  F = \lambda \Delta_0 + 0.5(1-\lambda) 
  \Delta_{c \pi} + 0.5 (1-\lambda) 
  \Delta_{-c \pi}, \hspace{12pt} 0 \leqslant
  \lambda < 0.5, \hspace{12pt}  0<c<1.
\end{equation*}
for $c>0$. Note that we set $\lambda < 0.5$ 
because if $\lambda \geqslant 0.5$ we would 
have $0 \leqslant \CMAD(F) < \pi/2$ by 
Lemma~\ref{Lemma:CMAD=0}. The case $c=0$ is not 
interesting (we obtain a point mass in 0, and 
$\CMAD(F)$ will be equal to zero again). For 
$c=1$, $F$ becomes 
$\lambda \Delta_0 + (1-\lambda) \Delta_{\pi}$ 
and again $\CMAD(F)=0$ for 
$\lambda < 0.5$ by Lemma~\ref{Lemma:CMAD=0}.

At these distributions we derive the expected 
value $\mathbb{E}(d(\theta,\phi))$ as a 
function of the generic angle $\theta$.
For the sake of simplicity, and without loss 
of generality by the symmetry of the 
distribution around 0, we investigate it only 
for $0 \leqslant \theta \leqslant \pi$. 
We have:
\begin{align*}
\mathbb{E}(d(\Theta,\theta)) &= 
 (1 - \lambda) d(\theta,0) + 0.5 \lambda 
 d(\theta,c\pi) + 0.5 \lambda d(\theta,-c\pi)  \\
&=  (1 - \lambda) \theta + 0.5 \lambda 
 (d(\theta,c\pi) + d(\theta,-c\pi)). 
\end{align*}
We will distinguish three cases. First, suppose 
$0 \leqslant \theta \leqslant -c \pi + \pi$. 
If so, we will have that $d(\theta,c\pi) + 
d(\theta,-c\pi) = 2 c \pi$. If instead 
$ -c \pi + \pi \leqslant \theta \leqslant c \pi$, 
then $d(\theta,c\pi) + d(\theta,-c\pi) = 
2((c\pi-\theta)+(\pi - c\pi)) = 
2 (\pi - \theta)$. Finally, if 
$c\pi \leqslant \theta \leqslant \pi$, then 
$d(\theta,c\pi)+d(\theta,-c\pi) =2(1-c)\pi$.
In brief, for 
$0 \leqslant \theta \leqslant \pi$ we 
can write:
\begin{equation} \label{eq:th3.4}
 \mathbb{E}(d(\Theta,\theta)) 
  = \lambda \theta+(1-\lambda) *
  \begin{cases}
    c\pi & 0 \leqslant \theta \leqslant
    (1-c)\pi \\
    \pi - \theta & (1-c)\pi < \theta 
    \leqslant c\pi \\
    (1-c) \pi &  c\pi < \theta 
    \leqslant \pi \,.
  \end{cases}
\end{equation}
This is a piecewise linear function that 
increases from 0 to $(1-c)\pi$, decreases 
from $(1-c)\pi$ up to $c\pi$, and increases 
again from $c\pi$ to $\pi$. As a consequence, 
extending the result to the whole circle, the 
median set $M_\theta$ can only be $\{ 0\}$ or 
$\{c\pi, -c\pi \}$. In the latter case, 
$\CMAD$ is not defined. 

We are interested in the first case. By 
definition of circular median, it occurs if 
$\mathbb{E}(d(\Theta,0)) < 
\mathbb{E}(d(\Theta,c\pi))$. That is, by 
Equation~\eqref{eq:th3.4}, if
\begin{equation}
    (1-\lambda) c \pi < \lambda  c\pi + 
    (1-\lambda) (1-c) \pi.
\end{equation}
This condition is satisfied whenever 
$c < (1-\lambda)/(2-3 \lambda)$. Hence, the 
circular random variable $\Theta'$ with
\begin{equation}
  F' := \lambda \Delta_0 + 0.5 (1-\lambda)
  \Delta_{c \pi} + 0.5 (1-\lambda) 
  \Delta_{-c \pi}, \hspace{12pt} 0 
  \leqslant \lambda < 0.5, \hspace{12pt}  
  0 < c < (1-\lambda)/(2-3 \lambda),
\end{equation}
has $M_{\theta'}=\{0\}$ by construction. 
Furthermore, it holds that $\CMAD(F')=c\pi$. 
This is because the distance distribution 
from $M_{\theta'}$ is a linear combination 
of two point mass distributions, one in $0$ 
(with probability $\lambda < 0.5$), and one 
in $c\pi$. As a consequence, given that 
$\lim_{\lambda \to 0.5}
(1-\lambda)/(2-3\lambda) = 1$,
\begin{equation*}
  \underset{\lambda, \hspace{2pt} c}{\sup}
  \hspace{4pt} \CMAD(F') = 
  \underset{\lambda, \hspace{2pt} c}{\sup} 
  \hspace{4pt} c\pi = \pi.
\end{equation*}
\vspace{6pt}

\noindent \textbf{Proof of 
Theorem~\ref{Th:CLMS_Characterization}}

($\Rightarrow$) If $\Theta \sim  H_{0.5}$, 
then the length of each arc with probability 
mass at least 0.5 is exactly equal to $\pi$. 
The result follows from 
Equation~\eqref{eq:CLMSDefinition2}.

($\Leftarrow$) Suppose that $\Theta$ has 
distribution $H \neq H_{0.5}$. By definition, 
it implies there exists at least one 
semicircle $HC^*$ such that 
$\int_{HC^*} dH(\theta) \neq 0.5$. 
If so, either $\int_{HC^*} dH(\theta) > 0.5$ 
and $\int_{\overline{HC^*}} dH(\theta) < 0.5$ 
or $\int_{HC^*} dH(\theta) < 0.5$ and 
$\int_{\overline{HC^*}} dH(\theta) > 0.5$ 
(here, ${\overline{HC^*}}$ is the complement 
to $\mathcal{S}$ of $HC^*$). In both cases, 
there will be at least an arc $A$ of length 
$\ell(A)=a < \pi$ such that 
$\int_{A} dH(\theta) = 0.5$. Hence, by 
Equation~\eqref{eq:CLMSDefinition2}, we
would have $\CLMS(H)=a/2 < \pi/2$, which 
contradicts the hypothesis.

\vspace{6pt}

\section*{B: Influence functions}

\noindent \textbf{Proof of 
  Theorem~\ref{Th.IFofCMADandCLMS}} 

We start with $\CMAD$. The population 
circular median of the 
uncontaminated distribution $F$ is unique 
due to the fact that $F$ is strictly 
unimodal, and equals the center of 
symmetry $0$.
The median arc length from the population
median is thus equal to 
$$\CMAD(F) = F^{-1}\Big(\frac{3}{4}\Big) - 
  F^{-1}\Big(\frac{1}{2}\Big) = 
  F^{-1}\Big(\frac{3}{4}\Big)\;.$$
We will denote $\q := F^{-1}(3/4)$ for
brevity.

Now let us contaminate $F$, yielding 
$\Fe := (1-\eps) F + \eps \Delta_\y$\;. 
The cumulative distribution function $\Fe$
is continuous on $[-\pi,\y[$, where $\Fe$
has the expression $\Fe(u) = (1-\eps)F(u)$.
We will denote the circular median of the 
contaminated distribution $\Fe$ 
as $t(\eps) := T(\Fe)$.
The uncontaminated distribution has 
$\eps=0$, so $t(0)=0$ as indicated in 
Figure~\ref{fig:Proof1_IFofCMAD}.
We also denote the $\CMAD$ of $\Fe$ by
$s(\eps)$, with $s(0)$ being the $\CMAD$ 
of the uncontaminated distribution $F$.
For $\eps > 0$ the figure changes, as both 
$t(\eps)$ and $s(\eps)$ can move. The 
endpoints of the $\CMAD$ interval 
now become $t(\eps) - s(\eps)$ and
$t(\eps) + s(\eps)$. 

\begin{figure}[!ht]
\begin{center}
\includegraphics[width=0.6\textwidth,trim=0cm 1.3cm 0cm 1.3cm, clip]
{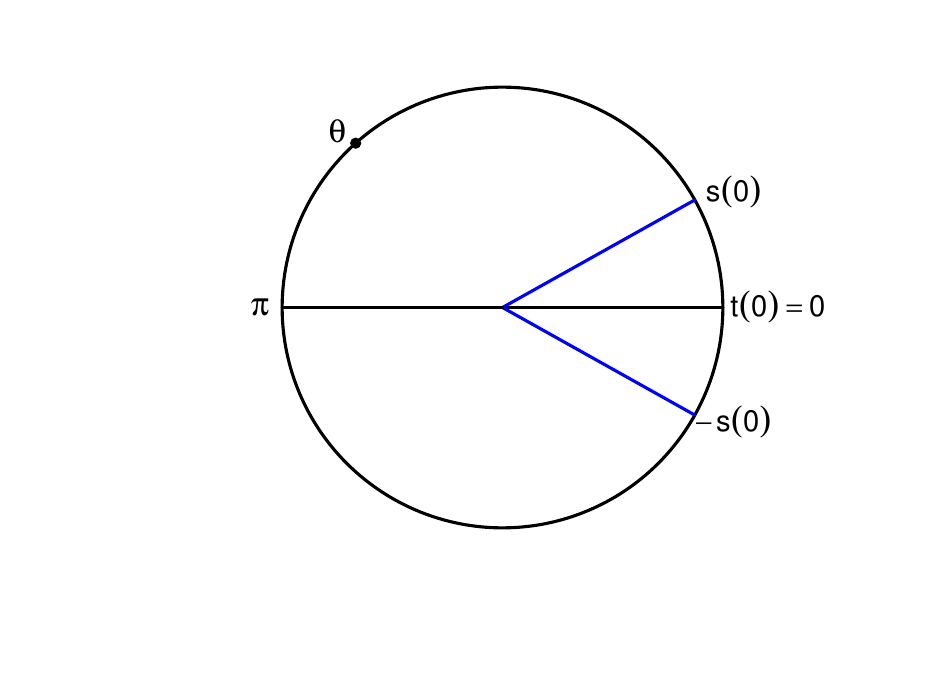}
\caption{Sketch of the first part of the proof
 for $\CMAD$, where $q < \theta < \pi$.}
\label{fig:Proof1_IFofCMAD}
\end{center}
\end{figure}

We first consider the situation where the
contaminating point $-\pi < \y < \pi$ lies 
outside the original interval, that is, 
$|\y| > s(0)$. Due to symmetry of $F$ 
we can assume without
loss of generality that $\y > s(0)$,
as in Figure~\ref{fig:Proof1_IFofCMAD}.
Note that for $\eps$ small enough we still
have $\y > t(\eps) + s(\eps)$ because
$\y > q = \lim_{\eps \downarrow 0} s(\eps)
= \lim_{\eps \downarrow 0} 
(t(\eps) + s(\eps))$.

By the definition of $\CMAD$ the mass of $\Fe$
between the endpoints must be one half, so
\begin{gather*}
  \Fe\big(t(\eps)+s(\eps)\big) -
  \Fe\big(t(\eps)-s(\eps)\big) = 
  \frac{1}{2} \\
  (1-\eps)F\big(t(\eps)+s(\eps)\big) -
  (1-\eps)F\big(t(\eps)-s(\eps)\big) = 
  \frac{1}{2} \;. 
\end{gather*}
Differentiating both sides of this equation
with respect to $\eps$ yields
\begin{gather*}
 -1\Big(F(s(0))-F(-s(0))\Big)
  + f(s(0))\Big(
  \frac{\partial t(\eps)} {\partial \eps} +
  \frac{\partial s(\eps)} {\partial \eps}
  \Big)\\
  - f(-s(0))\Big(
  \frac{\partial t(\eps)} {\partial \eps} -
  \frac{\partial s(\eps)} {\partial \eps}
  \Big) = 0 \\
  -\frac{1}{2} + f(s(0))\Big(2
    \frac{\partial s(\eps)} {\partial \eps}
  \Big) = 0 \\
  \IF(\y; \CMAD, F) =
  \frac{\partial s(\eps)} {\partial \eps} =
  \frac{1}{4 f(F^{-1}(\frac{3}{4})}
\end{gather*}
in which 
$\frac{\partial t(\eps)} {\partial \eps}$
has canceled.

We now consider the case where $|\y| < q$.
Figure~\ref{fig:Proof2_IFofCMAD} illustrates
this for $\y > 0$, but the situation 
$\y \leqslant 0$ is analogous. For $\eps$ 
small enough we have $\y < t(\eps) + s(\eps)$
because $ y < s(0) = \lim_{\eps \downarrow 0} 
(t(\eps) + s(\eps))$.

\begin{figure}[!ht]
\begin{center}
\includegraphics[width=0.6\textwidth,trim=0cm 1.3cm 0cm 1.3cm, clip]
{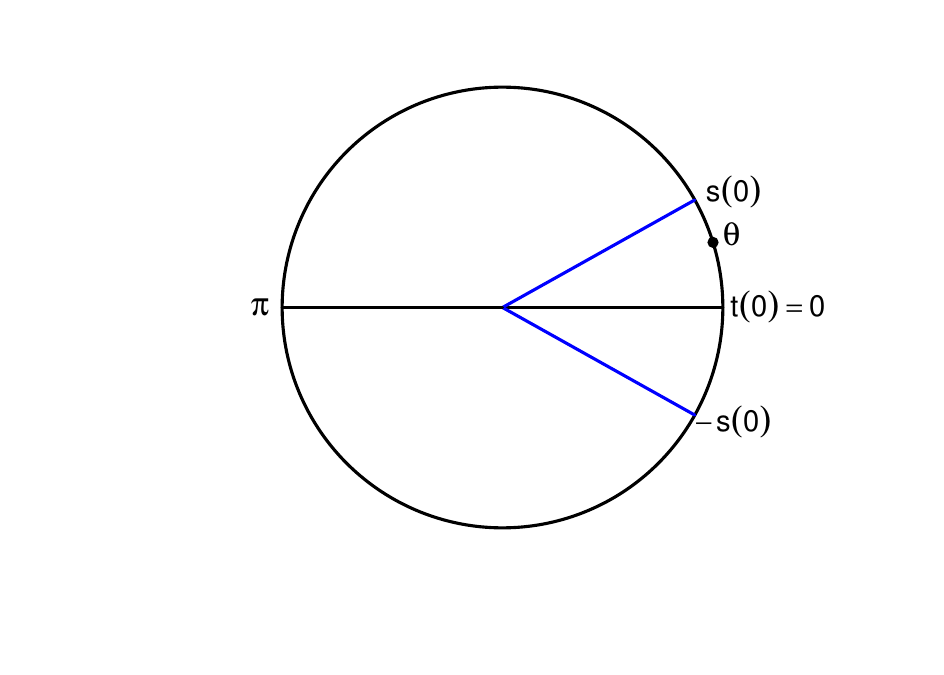}
\caption{Sketch of the second part of the proof 
for $\CMAD$, where $-q < \theta< q$.}
\label{fig:Proof2_IFofCMAD}
\end{center}
\end{figure}

By the definition of $\CMAD$ the mass of $\Fe$
between the endpoints is again one half, but
since $t(\eps) + s(\eps) > \y$ the expression
of $\Fe$ has an additional term $\eps$,
yielding
\begin{gather*}
  \Fe\big(t(\eps)+s(\eps)\big) -
  \Fe\big(t(\eps)-s(\eps)\big) = 
  \frac{1}{2} \\
  (1-\eps)F\big(t(\eps)+s(\eps)\big) 
  + \eps -
  (1-\eps)F\big(t(\eps)-s(\eps)\big) = 
  \frac{1}{2} \;. 
\end{gather*}
Differentiating both sides of this equation
with respect to $\eps$ yields
\begin{gather*}
 -1\Big(F(s(0))-F(-s(0))\Big) + 1
  + f(s(0))\Big(
  \frac{\partial t(\eps)} {\partial \eps} +
  \frac{\partial s(\eps)} {\partial \eps}
  \Big)\\
  - f(-s(0))\Big(
  \frac{\partial t(\eps)} {\partial \eps} -
  \frac{\partial s(\eps)} {\partial \eps}
  \Big) = 0 \\
  -\frac{1}{2} + 1 + f(s(0))\Big(2
    \frac{\partial s(\eps)} {\partial \eps}
  \Big) = 0 \\
  \IF(\y; \CMAD, F) =
  \frac{\partial s(\eps)} {\partial \eps} =
  \frac{-1}{4 f(F^{-1}(\frac{3}{4})}\;.
\end{gather*}
Combining the expressions for $|\y| > q$
and $|\y| < q$ yields the formula in
the theorem.\\

We now turn to the circular Least Median 
Spread. Due to the assumptions on $F$, 
the minimization of~\eqref{eq:functCLMS}
can be carried out by setting the first 
derivative of the objective to zero:
\begin{gather*}
 \frac{\partial}{\partial t}\,\frac{1}{2}
 \Big(F^{-1}(t + \frac{1}{2}) - 
    F^{-1}(t)\Big) = 0\\
 \frac{1}{f\big(F^{-1}(t + \frac{1}{2})\big)} -
      \frac{1}{f\big(F^{-1}(t)\big)} = 0\\
 F^{-1}(t + \frac{1}{2}) = F^{-1}(1 - t)\\
 t + \frac{1}{2} = 1 - t\\
 t = \frac{1}{4}
\end{gather*}
hence
$$\CLMS(F) = \frac{1}{2}\Big(F^{-1}
  \Big(\frac{3}{4}\Big) - 
  F^{-1}\Big(\frac{1}{4}\Big)\Big) = 
  F^{-1}\Big(\frac{3}{4}\Big)\;.$$
We again denote $\q := F^{-1}(3/4)$.

We now switch to the contaminated 
distribution 
$\Fe := (1-\eps) F + \eps \Delta_\theta$\;. 
We first consider the situation where
$\q < |\y| \leqslant \pi$\,. We 
assume w.l.o.g. that $\y > \q$, the case 
$\y<-\q$ being analogous by symmetry.
The situation is sketched in 
Figure~\ref{fig:Proof1_IFofCLMS}.
The distribution function $\Fe$
is continuous on $[-\pi,\y[$ where $\Fe$
has the expression $\Fe(u) = (1-\eps)F(u)$.
Therefore the total mass strictly to the
left of $\y$ is $P_\eps[-\pi < u < \y]
= (1-\eps)F(\y)$. For any 
$\eps < (F(\y)-F(q))/F(\y)$ it then holds
that $(1-\eps)F(\y) >
F(\y) - F(\y) + F(q) = F(q) = \frac{3}{4}$.
Therefore the total mass of $\Fe$ strictly
to the left of $\y$ is over $\frac{3}{4}$. 
This is where we will look for the 
shortest half.

\begin{figure}[H]
\begin{center}
\includegraphics[width=0.9\textwidth, trim=0cm 2.5cm 0cm 2cm, clip]
{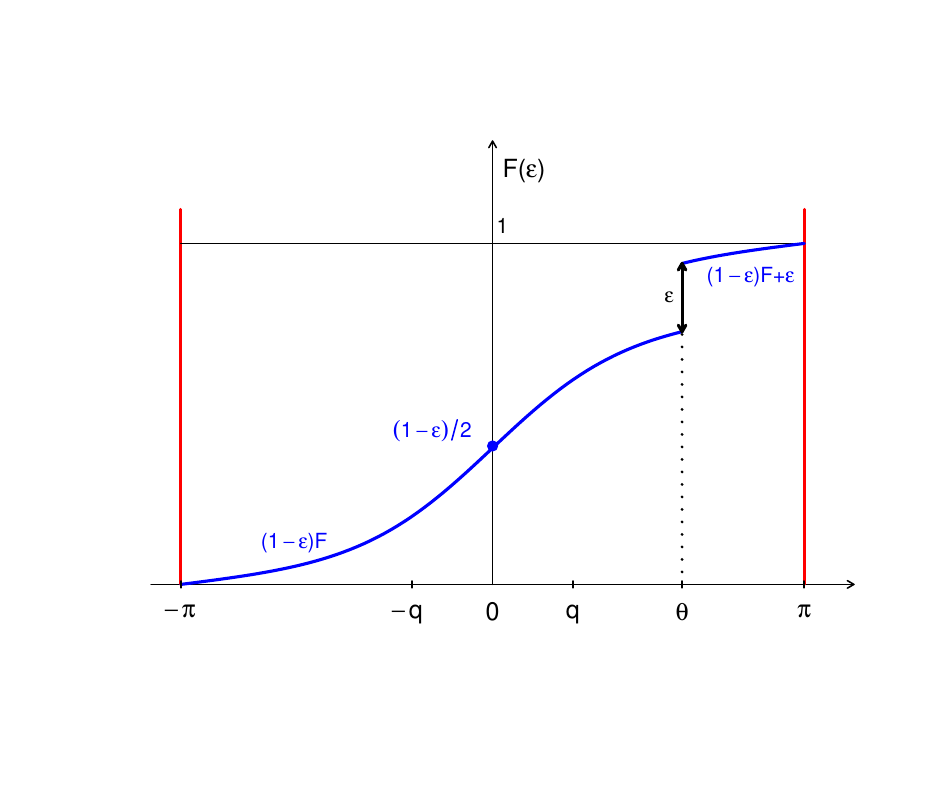}
\vspace{-5mm}
\caption{Sketch of the first part of the proof 
  for $\CLMS$, where $q < \theta \leq \pi$.}
\label{fig:Proof1_IFofCLMS}
\end{center}
\end{figure}

Minimizing~\eqref{eq:functCLMS} over this
region where $\Fe$ is smooth, we find
\begin{gather*}
 \frac{\partial}{\partial t}\,\frac{1}{2}
 \Big(\Fe^{-1}(t + \frac{1}{2}) - 
    \Fe^{-1}(t)\Big) = 0\\
 \frac{\partial}{\partial t}
 \Big(F^{-1}\Big((t + \frac{1}{2})/(1-\eps)
    \Big)- F^{-1}\Big(t/(1-\eps)\Big) = 0\\  
 f\Big(F^{-1}\big((t + \frac{1}{2})/(1-\eps)\big)
    \Big) = 
    f\Big(F^{-1}\Big(t/(1-\eps)\Big)\Big)\\   
 F^{-1}\Big((t + \frac{1}{2})/(1-\eps)\Big)
    = F^{-1}\Big(1- t/(1-\eps)\Big)\\      
 (t + \frac{1}{2})/(1-\eps)
    = (1 - \eps - t)/(1-\eps)\\     
  t = (1-2\eps)/4
\end{gather*}
so $\CLMS(\Fe)$ becomes
\begin{align*}
 \CLMS(\Fe) &= \frac{1}{2} \Big(
 \Fe^{-1}\Big((1-2\eps)/4 + \frac{1}{2}\Big) - 
    \Fe^{-1}\Big((1-2\eps)/4\Big) \Big)\\
 &= \frac{1}{2} \Big( 
    F^{-1}\Big((3-2\eps)/(4(1-\eps))\Big) - 
    F^{-1}\Big((1-2\eps)/(4(1-\eps))\Big)
    \Big)\;.
\end{align*}
Differentiation yields the influence function
in this $\y$:
\begin{align*}
 \IF(\y,\CLMS,F) &=
   \frac{\partial}{\partial \eps}\,
   \CLMS(\Fe)\Big|_{\eps = 0+}\\
 &=\frac{1}{2f(F^{-1}(3/4))}\,\frac{\partial}
   {\partial \eps} \,\frac{3 - 2\eps}
   {4(1-\eps)}\Big|_{\eps = 0+}\\ &\quad \quad -
   \frac{1}{2f(F^{-1}(3/4))}\,\frac{\partial}
   {\partial \eps} \,\frac{1 - 2\eps}
   {4(1-\eps)}\Big|_{\eps = 0+} \\
 &=\frac{1}{2f(F^{-1}(3/4))}\,\frac{\partial}
   {\partial \eps} \,\frac{1} {2(1-\eps)}
   \Big|_{\eps = 0+}\\
 &=\frac{1}{4\,f(F^{-1}(3/4))}\;.\\  
\end{align*}

\begin{figure}[H]
\begin{center}
\includegraphics[width=0.9\textwidth, trim=0cm 2.5cm 0cm 2cm, clip]
{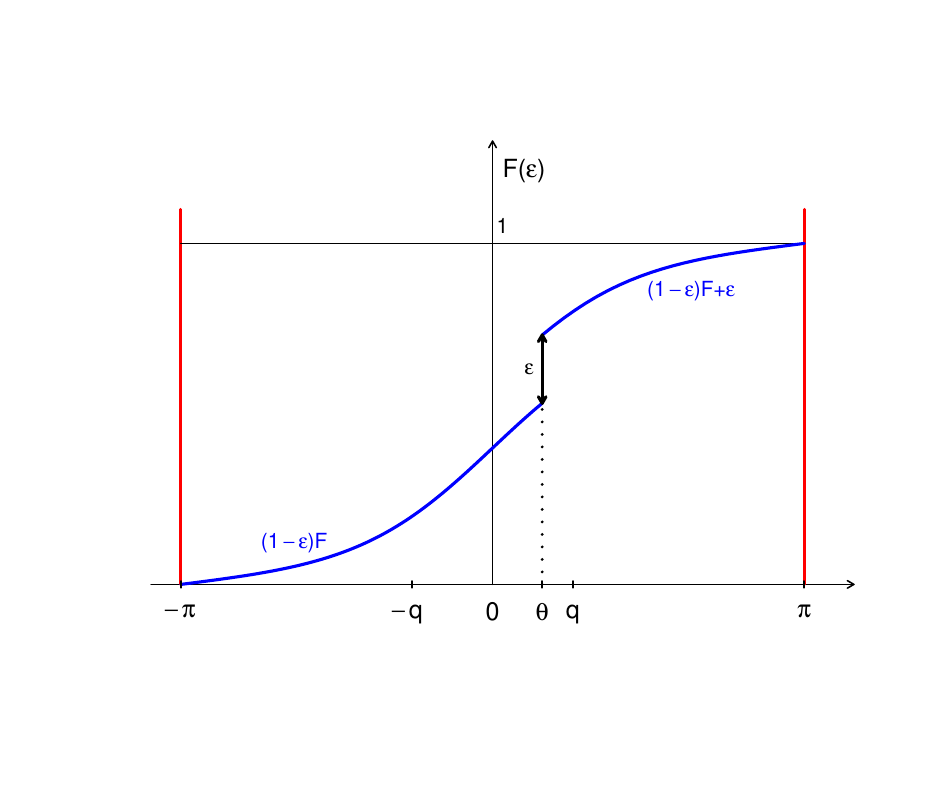}
\vspace{-5mm}
\caption{Sketch of the second part of the proof 
  for $\CLMS$, where $-q < \theta < q$.}
\label{fig:Proof2_IFofCLMS}
\end{center}
\end{figure}

We also have the situation where $-q < \y < q$
as sketched in Figure~\ref{fig:Proof2_IFofCLMS}.
The upper endpoint of the interval now lies in 
the region where 
$\Fe(u) = (1-\eps)F(u) + \eps$, 
so minimizing \eqref{eq:functCLMS} yields
\begin{gather*}
 \frac{\partial}{\partial t}\,\frac{1}{2}
 \Big(\Fe^{-1}(t + \frac{1}{2}) -
    \Fe^{-1}(t)\Big) = 0\\
 \frac{\partial}{\partial t}
 \Big(F^{-1}\Big((t + \frac{1}{2}-\eps)/(1-\eps)
    \Big)- F^{-1}\Big(t/(1-\eps)\Big) = 0\\  
 f\Big(F^{-1}\big((t+\frac{1}{2}-\eps)/(1-\eps)
   \big) \Big) = 
    f\Big(F^{-1}\Big(t/(1-\eps)\Big)\Big)\\   
 F^{-1}\Big((t+\frac{1}{2} -\eps)/(1-\eps)\Big)
    = F^{-1}\Big(1- t/(1-\eps)\Big)\\      
 (t + \frac{1}{2} - \eps)/(1-\eps)
    = (1 - \eps - t)/(1-\eps)\\     
  t = 1/4
\end{gather*}
so $\CLMS(\Fe)$ becomes
\begin{align*}
 \CLMS(\Fe)
 &= \frac{1}{2}\Big(\Fe^{-1}\Big(3/4\Big) - 
    \Fe^{-1}\Big(1/4\Big)\Big)\\
 &= \frac{1}{2}\Big(F^{-1}
    \Big((3/4-\eps)/(1-\eps)\Big) - 
    F^{-1}\Big((1/4)/(1-\eps))\Big)\Big)\\
 &= \frac{1}{2}\Big(F^{-1}
    \Big((3-4\eps)/(4(1-\eps))\Big) - 
    F^{-1}\Big(1/(4(1-\eps))\Big)\Big)\;.    
\end{align*}
Differentiation yields the influence function
in this $\y$:
\begin{align*}
 \IF(\y,\CLMS,F) &=
   \frac{\partial}{\partial \eps}\,
   \CLMS(\Fe)\Big|_{\eps = 0+}\\
 &=\frac{1}{2f(F^{-1}(3/4))}\,\frac{\partial}
   {\partial \eps} \,\frac{3 - 4\eps}
   {4(1-\eps)}\Big|_{\eps = 0+}\\ &\quad \quad -
   \frac{1}{2f(F^{-1}(3/4))}\,\frac{\partial}
   {\partial \eps} \,\frac{1}
   {4(1-\eps)}\Big|_{\eps = 0+} \\
 &=\frac{1}{2f(F^{-1}(3/4))}\,\frac{\partial}
   {\partial \eps} \,\frac{1-2\eps} {2(1-\eps)}
   \Big|_{\eps = 0+}\\
 &=\frac{-1}{4\,f(F^{-1}(3/4))}\;. 
\end{align*}
This ends the proof.\\

\noindent \textbf{Proof of 
Theorem~\ref{Th.IFofCLTS}} 
We will use the property that
the IF is symmetric because $F$ is, that is,
$\IF(-\y,S,F) = \IF(\y,S,F)$. 
Therefore we can compute the IF by
contaminating $F$ in a symmetric way, by 
putting
$\Fe := (1-\eps)F + \frac{\eps}{2}\Delta_\y
+\frac{\eps}{2}\Delta_{-\y}$ as illustrated
in Figure~\ref{fig:Proof1_IFofCLTS}.
Now the cumulative distribution function 
$\Fe$ is continuous on $[-\pi,-\y[$ where 
$\Fe(u) = (1-\eps)F(u)$, then jumps up 
$\frac{\eps}{2}$ at $-\y$,
continues as $(1-\eps)F(u)+
\frac{\eps}{2}$, jumps up again at $\y$,
and finally becomes $(1-\eps)F(u)+\eps$.
As before $\Fe(-\pi)=0$ and $\Fe(\pi)=1$.

\begin{figure}[H]
\begin{center}
\includegraphics[width=0.9\textwidth, trim=0cm 2cm 0cm 2cm, clip]
{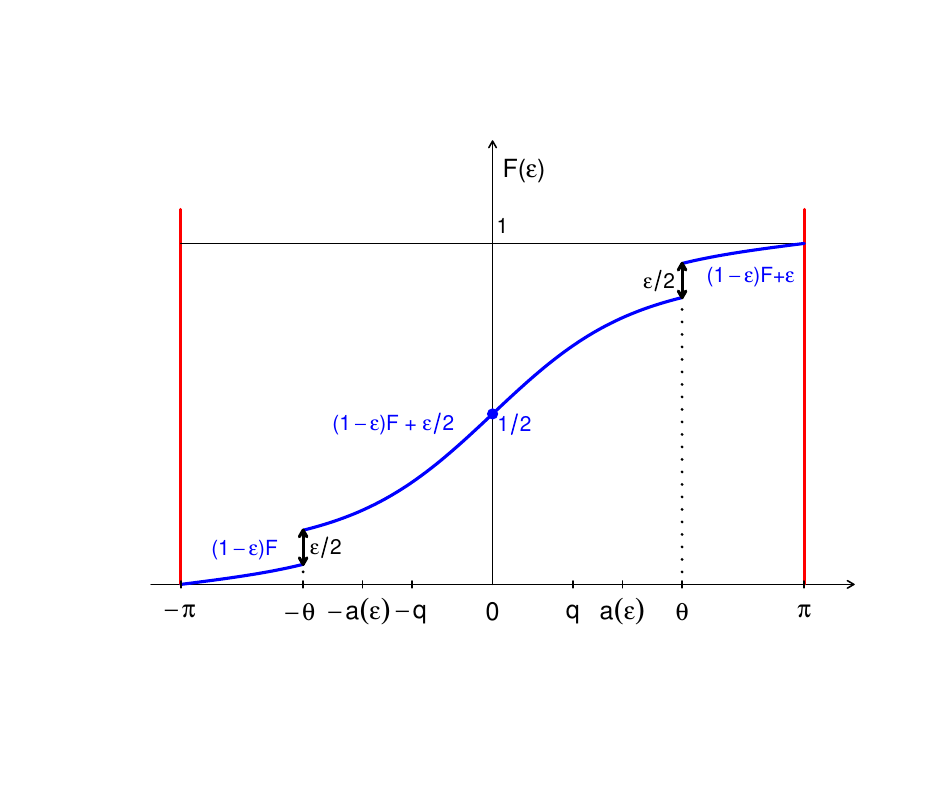}
\vspace{-5mm}
\caption{Sketch of the first part of the proof 
  for $\CLTS$, where $q < \theta \leq \pi$.}
\label{fig:Proof1_IFofCLTS}
\end{center}
\end{figure}

Let us start with the case where 
$|\y| > q$ where $q = F^{-1}(3/4)$. 
We can assume w.l.o.g.\ that $\y$ is 
positive, as in 
Figure~\ref{fig:Proof1_IFofCLTS}.
We will denote
the interval of mass $\frac{1}{2}$ and 
smallest circular variance by 
$[-a(\eps),a(\eps)]$. For small enough
$\eps$ we have $\y > a(\eps)$ because
we know that $\lim_{\eps \downarrow 0} 
a(\eps) = q < \y$\,. This implies that 
$[-a(\eps),a(\eps)]$ lies entirely in the
region where 
$\Fe(u) = (1-\eps)F(u) + \frac{\eps}{2}$
which yields
\begin{gather*}
   \frac{3}{4} = \Fe(a(\eps)) =
   (1-\eps)F(a(\eps)) + \frac{\eps}{2}\\
   (1-\eps)F(a(\eps)) =
   \frac{3-2\eps}{4}\\
   a(\eps) = F^{-1}\Big(
   \frac{3-2\eps}{4-4\eps}\Big)\,.
\end{gather*}
Now $S(\Fe)$ is given by the 
formula~\eqref{eq:CLTSformula} with
$\phi_1=-a(\eps)$ and $\phi_2=a(\eps)$.
By continuity of $\Fe$ in these two
points we know that $P(-a(\eps) \leqslant 
u \leqslant a(\eps)) = \frac{1}{2}$ (and 
not strictly larger than $\frac{1}{2}$).
Therefore~\eqref{eq:CLTSformula} becomes
\begin{align*} S(F_\eps) & = 
 2 \sqrt{\int_{-a(\eps)}^{a(\eps)}
 (1-\cos(u))d\Fe(u)} \\
 & = 2 \sqrt{0.5-c(\eps)}
\end{align*}
where $c(\eps)$ is the integral
$$c(\eps) := \int_{-a(\eps)}^{a(\eps)}
 \cos(u)d\Fe(u) = (1-\eps)
\int_{-a(\eps)}^{a(\eps)}\cos(u)f(u)du\,.$$
To obtain the IF we have 
to differentiate $S(\Fe)$ with respect to
$\eps = 0+$, yielding
\begin{align*}
\frac{\partial S(\Fe)}{\partial \eps} =&
   2\frac{\partial}{\partial \eps} 
   \sqrt{0.5 - c(\eps)}
 = \frac{2}{2\sqrt{0.5-c(0)}}\;
   \frac{\partial}{\partial \eps}(-c(\eps))\\
 =&\; \frac{2}{S(F)}\Big[
     \int_{-a(0)}^{a(0)}\cos(u)f(u)du
     - \cos(a(0)) f(a(0))
     \frac{\partial}{\partial \eps} a(\eps)\\
 &\; +\cos(-a(0)) f(-a(0)) \frac{\partial}
   {\partial \eps}\big(-a(\eps)\big)\Big]\\
 =&\; \frac{2}{S(F)}\Big[
    \int_{-q}^{q}\cos(u)f(u)du
   - 2\cos(q) f(q)\frac{\partial}{\partial \eps}
   F^{-1}\Big(\frac{3-2\eps}{4-4\eps}\Big)\Big]\\
 =&\; \frac{2}{S(F)}\Big[c(0)
   -2\cos(q) f(q) \frac{1}{f(F^{-1}(3/4))}
   \frac{\partial}{\partial \eps}
   \Big(\frac{3-2\eps}{4-4\eps}\Big)\Big] \\
 =&\; \frac{2}{S(F)}\Big[\,\frac{1}{2} - 
    \frac{S^2(F)}{4}
    -\frac{1}{2}\cos(F^{-1}(3/4)) \Big]\\ 
 =&\; \frac{1}{S(F)}\Big[\,1 - 
    \frac{S^2(F)}{2}
    -\cos(F^{-1}(3/4)) \Big]\;. 
\end{align*}

\begin{figure}[H]
\begin{center}
\includegraphics[width=0.9\textwidth, trim=0cm 2.5cm 0cm 2cm, clip]
{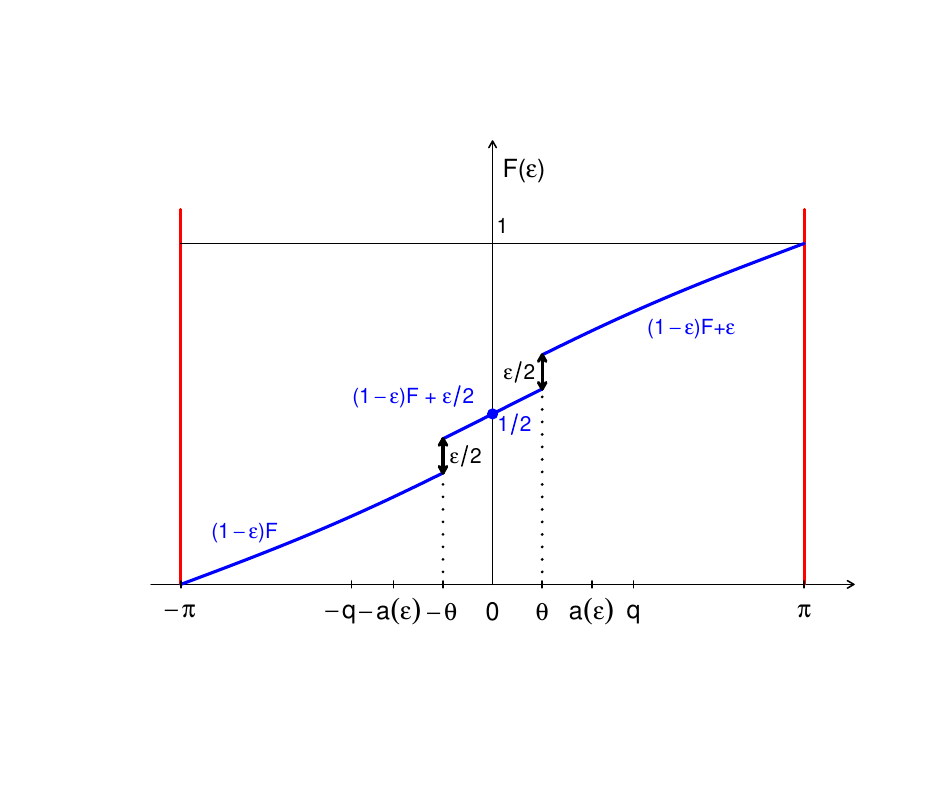}
\vspace{-5mm}
\caption{Sketch of the second part of the proof
  for $\CLTS$, where $-q < \theta < q$.}
\label{fig:Proof2_IFofCLTS}
\end{center}
\end{figure}

We now switch to the case where 
$|\y| < q$ as illustrated in
Figure~\ref{fig:Proof2_IFofCLTS}. For small 
enough $\eps$ we have $\y < a(\eps)$ because
$y < q = \lim_{\eps \downarrow 0} a(\eps)$. 
Therefore both upward jumps of $\Fe$ occur
inside $[-a(\eps),a(\eps)]$, so this time
\begin{gather*}
   \frac{3}{4} = \Fe(a(\eps)) =
   (1-\eps)F(a(\eps)) + \eps\\
   (1-\eps)F(a(\eps)) =
   \frac{3-4\eps}{4}\\
   a(\eps) = F^{-1}\Big(
   \frac{3-4\eps}{4-4\eps}\Big)\,.
\end{gather*}
Also here $S(\Fe)$ is twice the square root
of the circular variance on 
$[-a(\eps),a(\eps)]$, but the expressions
become a bit longer since the point masses 
in $-\y$ and $\y$ occur inside the 
interval, so now
\begin{align*}
  c(\eps) &= \int_{-a(\eps)}^{a(\eps)}
 \cos(u)d\Fe(u)\\
 &= (1-\eps)
\int_{-a(\eps)}^{a(\eps)}\cos(u)dF(u)
 + \frac{\eps}{2}\cos(\y)
 + \frac{\eps}{2}\cos(-\y)\\
  &= (1-\eps)
\int_{-a(\eps)}^{a(\eps)}\cos(u)dF(u)
 + \eps\cos(\y)\,.
\end{align*}
This time we obtain
\begin{align*}
  \frac{\partial S(\Fe)}{\partial \eps} =&\;
 \frac{\partial}{\partial \eps} 
 2\sqrt{0.5 - c(\eps)}
 = \frac{2}{2\sqrt{0.5-c(0)}}\;
 \frac{\partial}{\partial \eps} (-c(\eps))\\
 =&\; \frac{2}{S(F)}\Big[
    \int_{-a(0)}^{a(0)}\cos(u)f(u)du
    -\cos(a(0)) f(a(0))
    \frac{\partial}{\partial \eps} a(\eps)\\
 &\; +\cos(-a(0)) f(-a(0)) \frac{\partial}
   {\partial \eps}\big(-a(\eps)\big)
   -\frac{1}{2}\cos(\y)
   -\frac{1}{2}\cos(-\y) \Big]\\
 =&\; \frac{2}{S(F)}\Big[
    \int_{-q}^{q}\cos(u)f(u)du
   - 2\cos(q) f(q)\frac{\partial}{\partial \eps}
     F^{-1}\Big(\frac{3-4\eps}{4-4\eps}\Big)
   -\cos(\y) \Big]\\
 =&\; \frac{2}{S(F)}\Big[c(0)
   - 2\cos(q) f(q) \frac{1}{f(F^{-1}(3/4))}
   \frac{\partial}{\partial \eps}
   \Big(\frac{3-4\eps}{4-4\eps}\Big)
   -\cos(\y) \Big] \\
 =&\; \frac{2}{S(F)}\Big[\,\frac{1}{2} 
    - \frac{S^2(F)}{4}
    +\frac{1}{2}\cos(F^{-1}(3/4))
    -\cos(\y) \Big] \\
 =&\; \frac{1}{S(F)}\Big[\,1 
    - \frac{S^2(F)}{2}
    +\cos(F^{-1}(3/4))
    -2\cos(\y) \Big]\;.
\end{align*}
Combining the results for $|\y|>q$ and $|\y|<q$
yields the statement of the theorem.

\section*{C: Breakdown values}

\textbf{Proof of Lemma~\ref{Lemma:UpperboundBP}.}
We first show that a distribution $H^*$ exists 
such that the breakdown value is exactly equal 
to the right-hand expression in 
Equation~\eqref{eq:upper_bound_imp}. Then, we 
show there exists a second distribution $H^{**}$ 
where the breakdown value is lower, although 
not derivable in a closed form.

Consider the distribution 
$G^*:=(1-\eps)F + \eps H^*$, where 
$H^* := 0.5 \Delta_{\mu-\pi/2} +  
0.5 \Delta_{\mu+\pi/2}$. Note that, by the 
symmetry of $H^*$ with respect to $\mu$, and 
because $F$ has a unique mode in $\mu$ and 
density monotonically decreasing from $\mu$ 
to $\mu\pm \pi$, we have that the median of 
$G^*$ is unique and equals the median of $F$:
$\tilde{\mu}_{G^*}=\tilde{\mu}_{F}=\mu$. 

Consider the values of $\eps$ for which 
$\CMAD(G^*)=\pi/2$. By definition of $\CMAD$ 
(Equation~\ref{eq:CMAD_density}), 
$\CMAD(G^*)=\pi/2$ if 
\begin{equation}
\label{eq:condH*}
P_{G^*}(\mu-\pi/2 < \theta <  \mu+\pi/2) < 0.5
\end{equation}
and if
\begin{equation}
\label{eq:condH*2}
P_{G^*}(\mu-\pi/2 \leqslant \theta \leqslant 
  \mu +\pi/2) \geqslant 0.5,
\end{equation}
where the first condition ensures that 
there is no arc shorter than $\pi/2$ with 
probability content at least $0.5$ 
under $G^*$.

The second condition is always satisfied.
Given that both 
$P_{F}(\mu-\pi/2 \leqslant \theta 
\leqslant \mu +\pi/2) \geqslant 0.5$ and 
$P_{H^*}(\mu-\pi/2 \leqslant \Theta 
\leqslant \mu +\pi/2) \geqslant 0.5$, any 
linear combination of them is also 
$\geqslant 0.5$. Hence~\eqref{eq:condH*2} 
is satisfied for any $\eps \leqslant 0.5$.

For the first condition to be satisfied, 
a constraint applies. First note that, by
definition of $H^*$, $P_{H^*}(\mu-\pi/2 
< \Theta <  \mu+\pi/2) = 0$, and thus 
$$P_{G^*}(\mu-\pi/2 < \Theta <  \mu+\pi/2) =
(1-\eps) P_{F}(\mu-\pi/2 < \Theta <  
\mu+\pi/2) \hspace{12pt} \forall \eps.$$
Hence, condition~\eqref{eq:condH*} becomes
$$ (1-\eps) P_{F}(\mu-\pi/2 < \Theta <  
\mu+\pi/2)  < 0.5,$$
which is satisfied if 
$\eps > 1 - 0.5 [P_{F}(\mu-\pi/2 < 
\Theta <  \mu+\pi/2)]^{-1}$, 
so that
$$ \inf \{\eps: \CMAD(G^*)=\pi/2\} = 
  1 - 0.5 [P_{F}(\mu-\pi/2 < \Theta < 
  \mu+\pi/2)]^{-1}.$$

To complete the proof, we need to show that 
at least one distribution $H^{**}$ exists 
such that 
$\inf \{\eps: \CMAD(G^{**})=\pi/2\} 
< \inf \{\eps: \CMAD(G^*)=\pi/2\}$, where 
$G^{**}:=(1-\eps)F + \eps H^{**}$ (if not, 
then~\eqref{eq:upper_bound_imp} would not 
describe an upper limit but the breakdown
value itself).

Let $H^{**} := (1-\alpha) 
\Delta_{\tilde{\mu}_{G^{**}}-\pi/2} +  
\alpha \Delta_{\tilde{\mu}_{G^{**}}+\pi/2}, 
0 \leqslant \alpha \leqslant 1, \alpha 
\neq 0.5$. If so, then $G^{**}$ is an 
asymmetric distribution, and 
$\tilde{\mu}_{G^{**}} \neq \tilde{\mu}_{F}$, 
where the position of $\tilde{\mu}_{G^{**}}$ 
on the circle depends on the value of 
$\alpha$, and, most importantly, on the 
value of $\eps$ itself. As a consequence, 
no closed form can be derived to obtain 
the value of $\tilde{\mu}_{G^{**}}$.

On the other hand, $\tilde{\mu}_{G^{**}}
\in [\mu - \pi/2, \mu + \pi/2 ]$. First, 
note that $\tilde{\mu}_{G^{**}}$ cannot be 
antipodal to $\mu$ for any value of $\eps$ 
and $\psi$, because $\alpha \neq 0.5$. In 
addition, if $\alpha < 0.5$, then 
$\tilde{\mu}_{G^{**}} \in [\mu - \pi/2, \mu]$ 
$\forall \eps, \psi$, and vice versa.

Eventually, to get $\CMAD(G^{**})=\pi/2$, 
and hence to break down $\hpsi_{\uCMAD}$, 
in analogy with conditions~\eqref{eq:condH*}
and~\eqref{eq:condH*2}, we must have
\begin{equation}
P_{G^{**}}(\tilde{\mu}_{G^{**}}-\pi/2 < 
\Theta <  \tilde{\mu}_{G^{**}}+\pi/2) < 0.5
\end{equation}
and
\begin{equation}
P_{G^{**}}(\tilde{\mu}_{G^{**}}-\pi/2 
\leqslant \Theta \leqslant  
\tilde{\mu}_{G^{**}} +\pi/2) \geqslant 0.5.
\end{equation}

The second condition is always satisfied, 
given that both $P_{F}(\tilde{\mu}_{G^{**}}-
\pi/2 \leqslant \theta \leqslant  
\tilde{\mu}_{G^{**}} +\pi/2) \geqslant 0.5$, 
and $P_{H^{**}}(\tilde{\mu}_{G^{**}}-\pi/2 
\leqslant \Theta \leqslant 
\tilde{\mu}_{G^{**}} +\pi/2) \geqslant 0.5$, 
where the first holds because 
$\tilde{\mu}_{G^{**}} \in [\mu - \pi/2, 
\mu + \pi/2 ]$ and $F$ has $\geqslant 0.5$ 
probability for each semicircle that 
contains $\mu$.

For the first condition to be satisfied, 
because 
$P_{G^{**}}(\tilde{\mu}_{G^{**}}-\pi/2 < 
\Theta <  \tilde{\mu}_{G^{**}}+\pi/2) = 
(1-\eps) P_{F}(\tilde{\mu}_{G^{**}}-\pi/2
< \Theta <  \tilde{\mu}_{G^{**}}+\pi/2)$ 
by definition of $H^{**}$, we need 
$(1-\eps) P_{F}(\tilde{\mu}_{G^{**}}-\pi/2
< \Theta <  \tilde{\mu}_{G^{**}}+\pi/2) 
< 0.5$ which is in turn satisfied for 
$\eps > 1 - 0.5 [P_{F}(\tilde{\mu}_{G^{**}}
-\pi/2 < \Theta < 
\tilde{\mu}_{G^{**}}+\pi/2)]^{-1}$, so that
$$ 
\inf \{\eps: \CMAD(G^{**})=\pi/2\} = 
1 - 0.5 [P_{F}(\tilde{\mu}_{G^{**}}-\pi/2
< \Theta < 
\tilde{\mu}_{G^{**}} +\pi/2)]^{-1}.
$$

The proof is completed, as $\inf 
\{\eps: \CMAD(G^{**})=\pi/2\} < \inf 
\{\eps: \CMAD(G^{*})=\pi/2\}$ given that 
$P_{F}(\mu-\pi/2 < \Theta <  \mu+\pi/2)
> P_{F}(\tilde{\mu}_{G^{**}}-\pi/2 < 
\Theta <  \tilde{\mu}_{G^{**}}+\pi/2)$, 
because $\tilde{\mu}_{G^{**}} \neq \mu$. 
 
For both the von Mises and the wrapped 
normal distributions, $P_{F}(\mu-\pi/2 < 
\Theta <  \mu+\pi/2) > 0.5$ $\forall 
\psi >0$. Hence, the upper bound given 
in~\eqref{eq:upper_bound_imp} is always 
less than 0.5\,.\\

\textbf{Proof of Theorem~\ref{Theorem:ExpCMAD}.}
By definition of the explosion breakdown value, 
$\eps^*_{\uexp}(\hpsi_{\uCMAD}) = \inf \{\eps: 
\CMAD( (1-\eps)F + \eps H) = 0\}$. 
By the continuity of $\Theta$, and because of 
Lemma~\ref{Lemma:CMAD=0}, 
$\CMAD( (1-\eps)F + \eps H) = 0$ only 
if $H=\Delta_\theta$ for some $\theta$, and 
$\eps \geqslant 0.5$. The result follows. 

\section*{D: Asymptotic variances and 
   relative efficiencies}

In this section we study the asymptotic behavior of 
the robust dispersion functionals $\CMAD$, $\CLMS$, 
and $\CLTS$ at a parametric circular model. In a
given model, the functional is transformed to 
become Fisher consistent for a specific parameter,
here the $\kappa$ of the von Mises model and the
$\sigma$ of the wrapped normal. We will derive the
asymptotic variance of the transformed functional 
from that of the original functional.  

In the following theorems we compute the asymptotic 
variance of the functionals $\CMAD$, $\CLMS$, $\CLTS$, 
and their derived estimators.

\begin{theorem}[Asymptotic variance of $\CMAD$ and 
$\CLMS$.]\label{th.AVar_CMAD}
Let $F$ be a circular distribution with continuous 
density $f(\theta)>0$. Let $S(F)$ denote the $\CMAD$ 
or the $\CLMS$ functional and let $S_n=S(F_n)$ be 
its sample version based on the empirical distribution 
$F_n$. The influence function of $S$ is given by 
Theorem~\ref{Th.IFofCMADandCLMS}, and the asymptotic
variance of the estimator $S_n$ equals
\[
\AVar(S_n)
=\frac{1}{16\,f\!\left(F^{-1}(3/4)\right)^2}\;.
\]
\end{theorem}

\noindent \textit{Proof}
The asymptotic variance of an estimator $S_n$ is
given by
\[
\AVar(S_n) = \mathbb{E}_F\!\left[
\operatorname{IF}(\theta; S, F)^2 \right],
\]
where $\operatorname{IF}(\theta;S,F)$ is the 
influence function of $S$ at $F$. For $\CMAD$ 
or $\CLMS$, the influence function is
\[
\operatorname{IF}(\theta; S, F) 
= \frac{\operatorname{sign}\big(|\theta| -
   F^{-1}(3/4)\big)}{4\, f(F^{-1}(3/4))}\,.
\]
Squaring and taking expectation gives
\[
\mathbb{E}_F[\operatorname{IF}^2] 
= \int_{-\pi}^{\pi} \left(
\frac{\operatorname{sign}(|\theta| - 
 F^{-1}(3/4))}{4\, f(F^{-1}(3/4))}
\right)^2 f(\theta)\, d\theta
= \int_{-\pi}^{\pi} \frac{1}{16\, 
f(F^{-1}(3/4))^2} f(\theta)\, d\theta\,.
\]
Since $\int_{-\pi}^{\pi} f(\theta)\, 
d\theta = 1$, we obtain
\[
\AVar(S_n) = \frac{1}{16\, f(F^{-1}(3/4))^2}\,.
\]

\begin{theorem}[Asymptotic variance of $\CLTS$.]
Let $F$ be a circular distribution with continuous 
density $f(\theta)>0$ that is symmetric and strictly 
unimodal with mode $0$. Let $S(F)$ denote the $\CLTS$ 
functional and let $S_n = S(F_n)$ be its sample 
version based on the empirical distribution $F_n$. 
The influence function of $S$ is given by 
Theorem~\ref{Th.IFofCLTS}, and the asymptotic 
variance of $S_n$ equals
\[
\begin{aligned}
\mathrm{AVar}&(S_n) =
\frac{1}{S^2(F)}
\Bigg[
\left(1 - \frac{S^2(F)}{2} - 
\cos\bigl(F^{-1}(3/4)\bigr)\right)^2 \\[0.5ex]
&\quad + 4\left(1 - \frac{S^2(F)}{2} - 
\cos\bigl(F^{-1}(3/4)\bigr)\right)
\int_{F^{-1}(1/4)}^{F^{-1}(3/4)}
\bigl(\cos(F^{-1}(3/4)) - \cos\theta\bigr)\, 
dF(\theta) \\[0.5ex]
&\quad + 4 \int_{F^{-1}(1/4)}^{F^{-1}(3/4)}
\bigl(\cos(F^{-1}(3/4)) - \cos\theta\bigr)^2
\, dF(\theta) \Bigg].
\end{aligned}
\]
\end{theorem}

\noindent \textit{Proof}
The influence function of $S$ is
\[
\mathrm{IF}(\theta; S, F) = \frac{1}{S(F)}
  \Big[ 1 - 
  \frac{S^2(F)}{2} - \cos(F^{-1}(3/4)) + 
  2\, I\big(|\theta| < F^{-1}(3/4)\big)
  (\cos(F^{-1}(3/4)) - \cos\theta) \Big]\,,
\]
so the asymptotic variance equals
\begin{align*}
\mathrm{AVar}(S_n) &= \mathbb{E}_F
  \big[\mathrm{IF} (\theta; S, F)^2\big]\\
 &= \frac{1}{S^2(F)} \, \mathbb{E}_F \Big[
    \Big( 1 - \frac{S^2(F)}{2} - 
    \cos F^{-1}(3/4) \\
 &\quad\quad + 2\, I(|\theta|<F^{-1}(3/4))
   (\cos F^{-1}(3/4) - \cos\theta) \Big)^2
\Big].
\end{align*}
Expanding the square and using 
$F^{-1}(1/4) = - F^{-1}(3/4)$ 
allows the integrals to be written as
\begin{align*}
\mathrm{AVar}(S_n) &= \frac{1}{S^2(F)} 
 \Bigg[ \left(1 - \frac{S^2(F)}{2} - 
 \cos(F^{-1}(3/4))\right)^2 \\
&\quad + 4 \left(1 - \frac{S^2(F)}{2} - 
 \cos F^{-1}(3/4)\right) 
 \int_{F^{-1}(1/4)}^{F^{-1}(3/4)} 
 (\cos F^{-1}(3/4) - \cos\theta) \, 
 dF(\theta) \\
&\quad + 4 
 \int_{F^{-1}(1/4)}^{F^{-1}(3/4)} 
 (\cos F^{-1}(3/4) - \cos\theta)^2 \, 
 dF(\theta) \Bigg]\,.
\end{align*}

\begin{theorem}[Asymptotic variance of the 
transformed estimators.]
Let $S$ denote the $\CMAD$, $\CLMS$ or $\CLTS$ 
functional and let $\psi_S=\zeta(S)$ be its 
corresponding transformation.
Let $S_n=S(F_n)$ and $\hpsi_S=\zeta(S_n)$ be the 
corresponding estimator based on the empirical 
distribution $F_n$. If $\zeta$ is continuously 
differentiable at $S(F)$, then the asymptotic
variance of $\hpsi$ equals
\[
\AVar(\hpsi_S)
= \Big(\zeta'(S(F))\Big)^2\,\AVar(S_n).
\]
\end{theorem}

\begin{proof}
The influence function of $\psi_S$ is given by
\[
\operatorname{IF}(x;\psi_S,F)
=\zeta'(S(F))\,\operatorname{IF}(x;S,F).
\]
Substitution into the asymptotic variance 
formula
\[
\AVar(\hpsi_S)
=\int \operatorname{IF}(x;\psi_S,F)^2\,dF(x)
\]
yields the stated result.
\end{proof}

\begin{theorem}[Asymptotic variance of $\hkappa$ 
at the von Mises distribution.] Let $F$ be a von 
Mises distribution with concentration parameter 
$\kappa$, and let $\CSD$ be its population 
circular standard deviation that is related 
to $\kappa$ by
\[
\kappa = g(\CSD) =
A^{-1}\!\left(1-\frac{\CSD^2}{2}\right).
\]
Let $S_n$ be an estimator of $\CSD$ and define the 
induced estimator of $\kappa$ by $\hkappa = g(S_n)$. 
Then the asymptotic variance of $\hkappa$ equals
\begin{equation}
\label{eq:avar-kappa}
  \AVar(\hkappa) =
  \frac{\CSD^2}{A'(\kappa)^2}
  \AVar(S).
\end{equation}
\end{theorem}
\noindent \textit{Proof}
For the von Mises distribution
\[ \rho = A(\kappa) = 
  \frac{I_1(\kappa)}{I_0(\kappa)} \quad 
  \mbox{ and } \quad \CSD = \sqrt{2(1 - \rho)}\;,
\]
so $\rho = 1 - \CSD^2/2$\,. Solving for $\kappa$ 
as a function of $\CSD$, we obtain
\[ \kappa = 
  A^{-1}\!\left(1 - \frac{\CSD^2}{2}\right).
\]
We now define the function $g$ by 
$g(S) := A^{-1}(1 - S^2/2)$, so that 
$\kappa = g(S)$ and $\hkappa = g(S_n)$.
Since $A(\kappa)$ is continuously differentiable 
and strictly increasing, the function $g$ is 
differentiable at $\CSD$ with derivative
\[ g'(S) = 
  \frac{d}{dR}A^{-1}(R)\Big|_{R = A(\kappa)}
  \,\frac{d}{dS}\left(1 - \frac{S^2}{2}\right).
\]
Since
\[
\frac{d}{dR}A^{-1}(R)\Big|_{R = A(\kappa)} = 
\frac{1}{A'(\kappa)} \quad \mbox{ and } \quad
\frac{d}{dS}\left(1 - \frac{S^2}{2}\right) = -S,
\]
we obtain
\[
g'(S(F)) = -\frac{S(F)}{A'(\kappa)}
    = -\frac{\CSD}{A'(\kappa)}\;.
\]
We know that 
\[
\AVar(\hkappa) = (g'(S(F)))^2 \, \AVar(S_n).
\]
Therefore, the asymptotic variance of $\hkappa$ 
is indeed
\[
\AVar(\hkappa)
= \frac{\CSD^2}{\bigl(A'(\kappa)\bigr)^2}\AVar(S_n)\,.
\]

\begin{theorem}[Asymptotic variance of $\hsigma$ at 
the wrapped normal.]
Let $F$ be a wrapped normal distribution with scale 
parameter $\sigma$ and let $\CSD$ denote the circular 
standard deviation, that is related to $\sigma$ by 
\[
\sigma = g(\CSD) = \sqrt{-2 \log 
\left( 1 - \frac{\CSD^2}{2} \right)}.
\]
Let $S_n$ be an estimator of $\CSD$ and define the 
induced estimator of $\sigma$ by $\hsigma=g(S_n)$.
Then the asymptotic variance $\AVar(\hsigma)$ equals
\begin{equation}
\label{eq:avar-sigma}
\AVar(\hsigma) =
\frac{\CSD^2}{\sigma^2\Big(1-\frac{\CSD^2}{2}\Big)^2}
\AVar(S_n)\,.
\end{equation}
\end{theorem}

\noindent \textit{Proof}
For the wrapped normal distribution, the mean 
resultant length is
\[
\rho = e^{-\sigma^2/2}.
\]
Since $\rho$ is connected to CSD by 
$\CSD = \sqrt{2(1 - \rho)}$ we have
$\rho = 1 - \CSD^2/2$\,. This yields
\[
e^{-\sigma^2/2} = 1 - \frac{\CSD^2}{2}.
\]
Solving for $\sigma$ yields
\[
\sigma =
\sqrt{-2 \log\!\left(1 - \frac{\CSD^2}{2}\right)}.
\]
We now define the function $g$ by
\[
g(S) :=
\sqrt{-2 \log\!\left(1 - \frac{S^2}{2}\right)},
\]
so that $\sigma = g(S)$ and $\hsigma = g(S_n)$.
The function $g$ is continuously differentiable 
at $\CSD$ with derivative 
\[ g'(S) = 
\frac{d}{d\rho}\left(\sqrt{-2\log \rho}\right)
\Big|_{\rho = 1 - \frac{S^2}{2}} \,
\frac{d}{dS}\left(1 - \frac{S^2}{2}\right)\,.
\]
Since
\[
\frac{d}{d\rho}\left(\sqrt{-2\log \rho}\right)
= -\frac{1}{\rho \sqrt{-2\log \rho}}
\quad \mbox{ and } \quad
\frac{d}{dS}\left(1 - \frac{S^2}{2}\right) = -S,
\]
we obtain, using $\rho = 1 - \frac{S^2}{2}$ and 
$\sigma = \sqrt{-2\log \rho}$,
\[ g'(S) = 
 \frac{S}{\sigma\left(1 - \frac{S^2}{2}\right)}.
\]
Because we know that 
\[
\AVar(\hsigma) = ((\eta^{-1})'(S(F)))^2 
  \, \AVar(S_n)\,,
\]
the asymptotic variance of $\hsigma$ indeed 
equals
\[
\AVar(\hsigma) = \frac{\CSD^2}
{\sigma^2\left(1 - \frac{\CSD^2}{2}\right)^2}
\AVar(S_n)\,.
\]

\noindent \textbf{\large Efficiencies of 
\texorpdfstring{$\hpsi_{\text{CMAD}}$, 
$\hpsi_{\text{CLMS}}$, $\hpsi_{\text{CLTS}}$}
{psi hat CMAD, psi hat CLMS, psi hat CLTS}}\\ 

The efficiency of the proposed robust estimators 
is assessed through their asymptotic relative 
efficiency (ARE) with respect to the maximum 
likelihood estimator (MLE) under the assumed model.
For a parameter of interest $\psi$ (e.g., $\kappa$ 
for the von Mises distribution or $\sigma$ for the 
wrapped normal distribution), the ARE of an 
estimator $\hpsi$ is defined as
\[
\mathrm{ARE}(\hpsi)
=
\frac{
\mathrm{AVar}(\hpsi_{\mathrm{\MLE}})
}{
\mathrm{AVar}(\hpsi)
}.
\]

For the von Mises distribution, the maximum 
likelihood estimator \(\hkappa_{\mathrm{\MLE}}\) 
of the concentration parameter \(\kappa\) is 
asymptotically normal with variance
\[
\mathrm{AVar}(\hkappa_{\mathrm{\MLE}})
=
\frac{1}{I(\kappa)},
\]
where the Fisher information is given by
\[
I(\kappa)
=
\frac{1}{2}\left(1 + \frac{I_2(\kappa)}
 {I_0(\kappa)}\right)
-
\left(\frac{I_1(\kappa)}{I_0(\kappa)}\right)^2,
\]
with $I_0(\cdot)$, $I_1(\cdot)$, and $I_2(\cdot)$ 
denoting the modified Bessel functions of the 
first kind of orders $0$, $1$, and $2$.

Under the wrapped normal model with known $\mu$,
the maximum likelihood estimator of the scale 
parameter $\sigma$ is also asymptotically normal.
Specifically,
\[
\sqrt{n}\,(\hsigma_{\mathrm{\MLE}} - \sigma)
\xrightarrow{d}
\mathcal{N}\!\left(0,\, I(\sigma)^{-1}\right),
\]
where $I(\sigma)$ denotes the Fisher information 
that can be computed numerically.

\end{document}